\documentclass[twoside,11pt]{article}

\usepackage{article_style}

\usepackage{amsmath,amssymb}
\usepackage{bbm}
\usepackage{bm}
\usepackage{booktabs}
\usepackage{mathtools}
\usepackage{hyperref}
\hypersetup{
  pdftitle={Beyond Pretrends: A Discordance-Based Sensitivity Analysis for Difference-in-Differences},
  pdfauthor={},
  hidelinks,
  pdfcreator={LaTeX via pandoc}
}

\usepackage{cleveref}

\newcommand{\blind}{0}

\ifnum\blind=1
  \heading{00}{2026}{34}{}{}{}
  \ShortHeadings{Beyond Pretrends}{}
\else
  \heading{00}{2026}{34}{}{}{Thomas Leavitt}
  \ShortHeadings{Beyond Pretrends}{Leavitt}
\fi
\firstpageno{1}

\begin{document}

\title{Beyond Pretrends: A Discordance-Based Sensitivity Analysis for Difference-in-Differences\thanks{An earlier version of this paper circulated under the title ``Beyond Parallel Trends: Improvements on Estimation and Inference in the Difference-in-Differences Design.''}}

\ifnum\blind=1
\author{}
\else
\author{\name Thomas Leavitt \email thomas.leavitt@baruch.cuny.edu \\
       \addr Marxe School of Public and International Affairs \\
       Baruch College, City University of New York (CUNY) \\
       New York, NY, USA}
\fi

\maketitle

\begin{abstract}
In the canonical Difference-in-Differences design, the control group's post-treatment change serves as an imputation of the treated group's counterfactual change in the same period, an imputation justified by parallel trends. However, differences in group composition can produce between-group differences in how outcomes would evolve over time, rendering this imputation vulnerable to confounding. An alternative imputation --- such as one based on the treated group's pre-treatment change --- avoids such between-group confounding but introduces the risk of confounding from within-group temporal shifts. Ideally, both imputations, each vulnerable to different sources of confounding, would have concordant values, thereby yielding the same causal conclusions. When the imputations are discordant, conclusions under parallel trends hinge more critically on that assumption since alternative imputations would point to different results. Yet in these scenarios, existing pretrends-based sensitivity analyses can show low sensitivity because they ignore post-treatment deviations from pretrends in the control group. This paper therefore proposes a discordance-based sensitivity model in which parallel pretrends are necessary but not sufficient for low sensitivity. I formally justify this model in terms of the expected distance between the \textrm{ATT} under parallel trends and under alternative assumptions, weighted by the joint plausibility of those assumptions. I then provide a decision-theoretic rationale for benchmarking violations of parallel trends using the worst-case discordance between the parallel trends imputation and alternative imputations. Finally, I apply both pretrends- and discordance-based sensitivity models to assess how a labor supply shock influenced electoral support for apartheid-era policies in South Africa, showing how the two approaches yield different results.
\end{abstract}


\section{Introduction}

Difference-in-Differences (DID) is a canonical quasi-experiment, providing a powerful tool for inferring causal effects in observational studies. In a typical DID setup, one of two groups observed over time receives a treatment. To infer the average treatment effect on the treated group (the \textrm{ATT}), one compares the treated group's after-treatment minus before-treatment outcomes with the corresponding after-minus-before outcomes for a comparison group. Because the treated group's counterfactual change in outcomes under no treatment is unobservable, the comparison group's observed change serves as an imputation of this counterfactual change. The parallel trends assumption states that this imputation equals the treated group's counterfactual change in outcomes under no treatment. In most applications, researchers treat the plausibility of this assumption, which they assess via pre-treatment trends, as the key criterion for the credibility of the DID design.

Although parallel trends is virtually synonymous with the DID framework itself, much of the foundational literature points to a different basis for credibility. This literature on quasi-experiments \citep{shadishetal2002, cookcampbell1979} and, more broadly, observational studies \citep{cochran1965, cochran1972, rosenbaum2001} emphasizes that credibility often rests not on a single piece of evidence tied to a particular assumption. Rather, credibility depends on the alignment between multiple strands of evidence, each justified by different assumptions. As \citet[][p.~209]{rosenbaum2015} explains in ``Cochran's Causal Crossword,'' the strength of quasi-experiments arises from their capacity to
\begin{quote}
\textit{cultivate varied strands of evidence within a single study, each strand being weak on its own, each strand vulnerable in a different way, but with the several strands gaining in strength if they agree in appropriate ways.}
\end{quote}
In this view, quasi-experiments such as DID function as a crossword puzzle \citep{haack1995,haack1993}: An answer is credible when it aligns with multiple clues --- both across and down. An answer that fits one clue but contradicts another is less convincing than an answer that satisfies both.

This metaphor of the crossword puzzle is compelling in a DID setting, where the aim is to impute the treated group's after-minus-before average in the post-period. Consider two pieces of evidence, each vulnerable to a different source of confounding:

\begin{enumerate}

\item The first is the control group's after-minus-before average in the post-period. This serves as an imputation for the treated group's counterfactual change in outcomes in the same post-period. This imputation has the advantage of sharing the \textit{same time period} as the counterfactual. However, because this imputation comes from a group with a different composition, the imputation may misrepresent the counterfactual if the groups would evolve differently over time or are exposed to different post-period shocks, i.e., if parallel trends' approach fails.

\item The second is the treated group's after-minus-before average in the pre-period. This imputation avoids concerns about differences between groups because it uses data from the \textit{same group} as the counterfactual. But since this imputation reflects a different time period, the imputation relies on the assumption that the treated group would have evolved similarly across time. If this assumption fails, then the imputation may misrepresent the counterfactual.

\end{enumerate}

The first piece of evidence is vulnerable to \textit{between-group differences}, and the second to \textit{temporal shifts}. Ideally, both pieces of evidence would have similar values and thus support the same conclusion about the \textrm{ATT}. In that case, even if parallel trends were false and an alternative trend assumption were true, the resulting error from relying on parallel trends would be small. More generally, if an alternative trend assumption is at least somewhat plausible, concordance between the two pieces of evidence means one would expect inference based on parallel trends to have less error than it would if the two pieces of evidence were discordant. In short, concordance implies that causal conclusions are less sensitive to violations of parallel trends.

This straightforward intuition is not, however, reflected in existing sensitivity analyses to violations of parallel trends. As I go on to show, concordance implies parallel pretrends, but the converse is not true. Thus, pretrends-based sensitivity analyses can show high robustness to the failure of parallel trends, even if an alternative trend assumption, justified by a different strand of evidence, would yield a substantially different conclusion.

In this paper, I therefore propose a discordance-based sensitivity analysis. Since concordance implies parallel pretrends, this paper's sensitivity model will often give similar results to existing pretrends-based models. However, when conclusions diverge under alternative trend assumptions --- such as when pretrends are parallel but the control group's average post-period outcome deviates from its pretrend --- this paper's model will indicate greater sensitivity to violations of parallel trends.

In what follows, I develop a discordance-based sensitivity analysis that builds on existing frameworks. This development begins with a formal setup in \Cref{sec: setup}, a review of the pretrends-based model in \Cref{sec: pretrends model}, and the motivation for alternatives in \Cref{sec: motivation}, followed by the introduction of the discordance-based approach in \Cref{sec: discordance-based model}. This approach bounds violations of parallel trends based on the discordance between the parallel trends imputation and alternative possible imputations for the treated group's counterfactual. As I show in \Cref{ssec: concordance}, the model's use of discordance can be formally justified in terms of the expected distance between the \textrm{ATT} under parallel trends and the \textrm{ATT} under alternative identification assumptions, evaluated with respect to the joint plausibility of those assumptions (including that of non-identification). I also show that concordance implies parallel pretrends, though the converse does not hold. \Cref{ssec: weighting} offers a decision-theoretic justification for the model's use of the worst-case discordance between the parallel trends imputation and any alternative imputation. \Cref{sec: multiple post periods} then generalizes the argument --- developed up to that point for a single post-treatment period --- to settings with multiple post-periods. \Cref{sec: application} applies the discordance-based sensitivity analysis to a case study examining how a labor supply shock in apartheid-era South Africa affected electoral support for apartheid policies. The final section concludes.

\section{Setup} \label{sec: setup}

Suppose two groups, a treated group ($G = 1$) and a comparison group ($G = 0$; hereafter also called the control group), where the index set of groups is denoted by $\mathcal{G} \coloneqq \{0, 1\}$. Also suppose $t = 1, \ldots , T$ periods, of which $T$ is the only post-treatment period. That is, between periods $T - 1$ and $T$, the treated group is exposed to treatment and the comparison group is not. Denote the index set of time periods by $\mathcal{T} \coloneqq \{1, \ldots , T\}$. For all $(g,t) \in \mathcal{G} \times \mathcal{T}$, let the treatment indicator for group $g$ in period $t$ be $Z_{g,t} \coloneqq \mathbbm{1}\{g = 1\} \mathbbm{1}\{t = T\}$, where $\mathbbm{1}\{\cdot\}$ is the indicator function that equals $1$ if its argument is true and $0$ if not. For the treated group, $Z_{g,t} = 0$ for all $t < T$ and $Z_{g,T} = 1$. For the control group, $Z_{g,t} = 0$ for all periods.

To define the causal target in terms of potential outcomes, let $Y_{g,t}(0)$ denote the untreated potential outcome for any group-period $(g,t) \in \mathcal{G} \times \mathcal{T}$. Also let $Y_{1,T}(1)$ denote the treated potential outcome for group $G = 1$ in the post-treatment period, $T$. The causal target is the ATT, denoted by $\tau$,
\begin{align} \label{eq: ATT I}
\tau \coloneqq \bar{Y}_{1,T}(1) - \bar{Y}_{1,T}(0),
\end{align} 
where, for any random quantity $W$, I write $\bar{W}$ to denote the expectation (mean) of $W$ with respect to the population-level joint cumulative distribution function (CDF). For example, $\bar{Y}_{1,T}(1)$ is the expected treated potential outcome for group $G = 1$ in period $T$.

This basic setup features two groups and a single post-treatment period, but it generalizes easily to more complex settings. In staggered-adoption settings with variation in treatment timing, \citet{callawaysantanna2021} show that the analysis can be decomposed into a set of simple DID comparisons between each timing-defined group and a never-treated or not-yet-treated group. The same basic setup also extends naturally to settings with multiple post-treatment periods, as detailed in \Cref{sec: multiple post periods}. In such cases, there are multiple post-treatment \textrm{ATT}s, and the target parameter can be any linear combination of them, with a common choice being the average \textrm{ATT} across all post-treatment periods \citep[][p.~2561]{rambachanroth2023}.

Returning to the basic setting with two groups and one post-period, I first state the standard ``consistency'' assumption of \citet{hernanrobins2020}, which enables expressing the \textrm{ATT} in terms of observable quantities:
\begin{assumption} \label{assumption: consistency}
For all $(g,t) \in \mathcal{G} \times \mathcal{T}$,
\begin{equation} \label{eq: consistency}
Y_{g,t} = Z_{g,t} Y_{g,t}(1) + \left(1 - Z_{g,t}\right)Y_{g,t}(0).
\end{equation}
\end{assumption}
\noindent Assumption \ref{assumption: consistency} implies that the observed outcome in each period equals the potential outcome under the treatment status actually received. This assumption rules out treatment anticipation and interference, ensuring that each unit's observable outcome does not reflect any unit's future treatment status or the treatment status of other units at the same time. Moreover, by defining a single treated and a single untreated potential outcome, this setup assumes a single, well-defined version of the treatment.

Under Assumption \ref{assumption: consistency}, the \textrm{ATT} is
\begin{equation}\label{eq: ATT II}
\tau = \bar{Y}_{1,T} - \bar{Y}_{1,T}(0).
\end{equation}
Assumption \ref{assumption: consistency} justifies substituting the treated group's observable outcome in the post-period for the treated group's treated potential outcome in the post-period. Hence, the expression for the \textrm{ATT} in \eqref{eq: ATT II} contains $\bar{Y}_{1,T}$ instead of $\bar{Y}_{1,T}(1)$.

The \textrm{ATT} in \eqref{eq: ATT II} is now composed of an observable quantity, $\bar{Y}_{1,T}$, and an unobservable one, $\bar{Y}_{1,T}(0)$. To fully express the \textrm{ATT} in only observable quantities, it is helpful to first express the \textrm{ATT} in \eqref{eq: ATT II} as
\small
\begin{equation} \label{eq: ATT III}
\begin{split}
\tau & = \bar{Y}_{1,T} - \bar{Y}_{1,T-1}(0) - \left(\bar{Y}_{0,T}(0) - \bar{Y}_{0,T-1}(0)\right) \\ 
& - \left[\bar{Y}_{1,T}(0) - \bar{Y}_{1,T-1}(0) - \left(\bar{Y}_{0,T}(0) - \bar{Y}_{0,T-1}(0)\right)\right].
\end{split}
\end{equation}
\normalsize
This expression for the \textrm{ATT} in \eqref{eq: ATT III} reduces to the canonical \textrm{DID},
\begin{align} \label{eq: DID I}
\textrm{DID} & \coloneqq \bar{Y}_{1,T}(1) - \bar{Y}_{1,T-1}(0) - \left(\bar{Y}_{0,T}(0) - \bar{Y}_{0,T-1}(0)\right),
\end{align}
minus period $T$'s Difference-in-Trends, where the Difference-in-Trends for any $t \geq 2$ is
\begin{align} \label{eq: diff-in-trends I}
\delta_t & \coloneqq \bar{Y}_{1,t}(0) - \bar{Y}_{1,t-1}(0) - \left(\bar{Y}_{0,t}(0) - \bar{Y}_{0,t-1}(0)\right).
\end{align}
Therefore, the \textrm{ATT} in \eqref{eq: ATT III} is 
\begin{align} \label{eq: ATT IV}
\tau = \textrm{DID} - \delta_T.
\end{align}

The canonical point identification assumption about $\delta_T$ is parallel trends, which, for any $t \geq 2$, is
\begin{align} \label{eq: parallel trends}
\delta_t & = 0.
\end{align}
Under the assumption that $\delta_T = 0$, the ATT in \eqref{eq: ATT IV} equals the DID in \eqref{eq: DID I}, making the unobservable \textrm{ATT} point identified in terms of observable quantities. However, any point assumption about the value of the counterfactual $\delta_T$, not just $\delta_T = 0$, allows one to recover the \textrm{ATT} by subtracting the assumed value of $\delta_T$ from the \textrm{DID}.

To simplify this framework going forward, define a detrended outcome for group $g$ in period $t \geq 2$ as
\begin{align} \label{eq: detrended outcome I}
\epsilon_{g,t} & \coloneqq Y_{g,t} - \bar{Y}_{g,t-1},
\end{align}
which is simply the observable outcome for group $g$ in period $t$ minus the mean (expectation) of group $g$'s observable outcome in the previous period ($t - 1$). The mean detrended outcome, $\bar{\epsilon}_{g,t}$, therefore represents the period-to-period change in group $g$'s average outcome. If the underlying outcome series for group $g$ exhibits an upward or downward trend, the sequence of detrended outcomes in the pre-treatment periods, $ \{\epsilon_{g,t}\}_{t=2}^{T-1}$, will generally reflect those changes. This transformation adjusts the level (i.e., the average outcome at each point in time) of the series but does not impose any shape restrictions --- such as a constant slope or flat trend --- on the detrended outcomes.

Using the detrended outcomes, the setup thus far admits an equivalent representation. By combining the definition in \eqref{eq: detrended outcome I} with Assumption \ref{assumption: consistency}, the ATT in \eqref{eq: ATT I} can be rewritten exactly in terms of the detrended potential outcomes. The algebraic steps below make this equivalence clear.

First, recall that $Z_{g,t} \coloneqq \mathbbm{1}\{g = 1\} \mathbbm{1}\{t = T\}$, so that $Z_{g,t} = 0$ for all $t < T$. In those periods, Assumption \ref{assumption: consistency} implies that the detrended outcome in \eqref{eq: detrended outcome I} can be expressed as $Y_{g,t} - \bar{Y}_{g,t-1}(0)$. Substituting the expression for $Y_{g,t}$ from \eqref{eq: consistency} into \eqref{eq: detrended outcome I} then yields $\epsilon_{g,t} = \left[Z_{g,t} Y_{g,t}(1) + \left(1 - Z_{g,t}\right)Y_{g,t}(0)\right] - \bar{Y}_{g,t-1}(0)$. This expression is algebraically equivalent to 
\begin{align} \label{eq: detrended outcome II}
\epsilon_{g,t} & = Z_{g,t} \left(Y_{g,t}(1) - \bar{Y}_{g,t-1}(0)\right) + \left(1 - Z_{g,t}\right)\left(Y_{g,t}(0) - \bar{Y}_{g,t-1}(0)\right),
\end{align}
which equals $Y_{g,t}(1) - \bar{Y}_{g,t-1}(0)$ when $Z_{g,t} = 1$ and $Y_{g,t}(0) - \bar{Y}_{g,t-1}(0)$ when $Z_{g,t} = 0$. I refer to these as the detrended potential outcomes, defined as
\begin{align} \label{eq: detrended potential outcome}
\epsilon_{g,t}(z) & \coloneqq Y_{g,t}(z) - \bar{Y}_{g,t-1}(0),
\end{align}
for $z \in \{0, 1\}$, corresponding to the detrended untreated ($z = 0$) and treated ($z = 1$) potential outcomes, respectively. Using this definition in \eqref{eq: detrended potential outcome}, the detrended outcome in \eqref{eq: detrended outcome II} can be written compactly as 
\begin{align} \label{eq: detrended outcome III}
\epsilon_{g,t} = Z_{g,t} \epsilon_{g,t}(1) + (1 - Z_{g,t})\epsilon_{g,t}(0),
\end{align}
analogous to the original expression of the ``consistency'' assumption in \eqref{eq: consistency}.

It now follows directly that the \textrm{ATT} in \eqref{eq: ATT I}, defined in terms of the raw potential outcomes, is algebraically equivalent to its detrended form, $\bar{\epsilon}_{1,T}(1) - \bar{\epsilon}_{1,T}(0)$:
\begin{align} 
\tau & \coloneqq \bar{Y}_{1,T}(1) - \bar{Y}_{1,T}(0)  = \bar{Y}_{1,T}(1) - \bar{Y}_{1,T-1} - \left(\bar{Y}_{1,T}(0) - \bar{Y}_{1,T-1}\right) \notag \\ 
& = \bar{\epsilon}_{1,T}(1) - \bar{\epsilon}_{1,T}(0). \label{eq: ATT V}
\end{align} 
Similarly, the \textrm{DID} in \eqref{eq: DID I} and the Difference-in-Trends in \eqref{eq: diff-in-trends I} can both be expressed in terms of detrended potential outcomes as
\begin{align}
\textrm{DID} & \coloneqq \bar{Y}_{1,T}(1) - \bar{Y}_{1,T-1}(0) - \left(\bar{Y}_{0,T}(0) - \bar{Y}_{0,T-1}(0)\right) = \bar{\epsilon}_{1,T}(1) - \bar{\epsilon}_{0,T}(0) \label{eq: DID II}
\end{align}
and
\begin{align}
\delta_t & \coloneqq \bar{Y}_{1,t}(0) - \bar{Y}_{1,t-1}(0) - \left(\bar{Y}_{0,t}(0) - \bar{Y}_{0,t-1}(0)\right) = \bar{\epsilon}_{1,t}(0) - \bar{\epsilon}_{0,t}(0). \label{eq: diff-in-trends II}
\end{align}
As in the earlier formulation based on the raw potential outcomes, it remains the case under the detrended-outcome notation that the \textrm{ATT} is the \textrm{DID} minus period $T$'s Difference-in-Trends,
\begin{align} \label{eq: ATT VI}
\textrm{DID} - \delta_T = \bar{\epsilon}_{1,T}(1) - \bar{\epsilon}_{0,T}(0) - \left[\bar{\epsilon}_{1,T}(0) - \bar{\epsilon}_{0,T}(0)\right] = \bar{\epsilon}_{1,T}(1) - \bar{\epsilon}_{1,T}(0) = \tau.
\end{align}

The parallel trends assumption states that $\delta_T = 0$, i.e., that the unobservable $\bar{\epsilon}_{1,T}(0)$ equals the observable $\bar{\epsilon}_{0,T}(0)$. To distinguish the assumption itself from the specific counterfactual imputation the assumption justifies (regardless of whether it actually holds), I refer to $\bar{\epsilon}_{0,T}$ as the parallel trends imputation for $\bar{\epsilon}_{1,T}(0)$. I denote the corresponding imputed \textrm{ATT} by $\tau_{0,T}$, where the subscript $(0,T)$ indicates that the counterfactual $\bar{\epsilon}_{1,T}(0)$ is imputed using the comparison group's detrended outcome in period $T$, $\bar{\epsilon}_{0,T}$.

\citet{leavitthatfield2025} extend the detrended-outcome approach to a much broader class of prediction functions, introducing a general point identification requirement they term equal expected prediction errors. They show that DID can be viewed as a special case of a more general two-stage procedure: The first stage uses pre-treatment data to \textit{predict} untreated post-treatment outcomes, and the second stage uses the comparison group's observable post-treatment prediction errors to \textit{correct} the treated group's predictions \citep[see][p.~1827]{leavitthatfield2025}. In standard DID, the parallel trends assumption implies equal expected prediction errors when the prediction function is $\bar{Y}_{g,t-1}$, which is precisely the detrended-outcome formulation in \eqref{eq: detrended outcome I}. This broader framework encompasses many other prediction models, including sequential DID, two-way fixed effects, and models with unit- or group-specific time trends, among others. The key implication is that the reasoning developed here for a simple prediction model extends naturally to a much wider class of models for predicting how outcomes evolve over time.

Finally, to complete the formal setup, let $\mathcal{V} \subseteq \{2, \ldots, T - 1\}$ denote the set of pre-treatment validation periods selected for analysis, with each $v \in \mathcal{V}$ indexing a single period. The set $\mathcal{V}$ can be any nonempty subset of $\{2, \ldots, T - 1\}$,  such as $\{2\}$, $\{T - 1\}$, or the full set $\{2, \ldots, T - 1\}$. If one imputes the counterfactual using the detrended outcome from a particular period $v \in \mathcal{V}$ and group $g \in \mathcal{G}$, $\bar{\epsilon}_{g,v}$, the resulting imputed \textrm{ATT} is denoted $\tau_{g,v}$.

\section{Pretrends-based Sensitivity Analysis} \label{sec: pretrends model}

Researchers cannot directly evaluate whether parallel trends (i.e., $\delta_T = 0$) holds since untreated potential outcomes in period $T$ are unobservable. However, researchers can directly evaluate parallel trends in periods before treatment, i.e., whether $\delta_t = 0$ for $2 \leq t < T$. Hence, researchers typically assess the plausibility of parallel trends by assessing whether pretrends are parallel. Presumably in response to this widespread practice, \citet{rambachanroth2023} derive a general sensitivity analysis framework that imposes restrictions on the extent to which parallel trends is violated based on violations in the pre-period. The framework is general and enables  a wide class of set restrictions on the unobservable value of $\delta_T$.

A leading set restriction on $\delta_T$ that \citet{rambachanroth2023} propose (their ``relative magnitudes'' restriction, which builds on the sensitivity analysis in \citet{manskipepper2018}) stipulates that violations of parallel trends are no greater than a magnitude, $M \geq 0$, of the greatest absolute value of the Difference-in-Trends in the pre-period. Formally, one can write this sensitivity analysis as one in which the counterfactual Difference-in-Trends ($\delta_T$) lies in a compact set given by
\begin{align} \label{eq: max set iden}
\left\{\delta_T: \, \abs{\delta_T} \leq M \max \limits_{v \in \mathcal{V}} \abs{\delta_v} \right\}.
\end{align}

The sensitivity parameter $M \geq 0$ controls how tightly one constrains the identification assumption in terms of the greatest absolute violation of parallel trends in the pre-period. Point identification of the \textrm{ATT} holds under $M = 0$, which implies parallel trends. Otherwise, so long as $\max_{v \in \mathcal{V}} \abs{\delta_v} > 0$, the \textrm{ATT} is set-identified rather than point-identified. In particular, under the set restriction in \eqref{eq: max set iden}, the \textrm{ATT} belongs to the following set:
\begin{align}
\left[\textrm{DID} - M \max \limits_{v \in \mathcal{V}} \abs{\delta_v}, \, \textrm{DID} + M \max \limits_{v \in \mathcal{V}} \abs{\delta_v} \right].
\end{align}

To bound the post-period violation of parallel trends, one might want to use the average pre-period violation --- or another function of the pre-period violations --- instead of the maximum, as in \eqref{eq: max set iden}. To accommodate these possibilities, I embed the leading sensitivity analysis of \citet{rambachanroth2023} as a special case of a sensitivity analysis in which
\begin{align} \label{eq: weighted set iden}
\left\{\delta_T: \, \abs{\delta_T} \leq M \sum \limits_{v \in \mathcal{V}} w_v \abs{\delta_v} \right\},
\end{align}
where $\{w_{v}\}_{v \in \mathcal{V}}$ are weights over the set $\mathcal{V}$ defined as
\begin{equation} \label{eq: weights I}
\begin{split}
\left\{w_{v}\right\}_{v \in \mathcal{V}} \in \left\{\bm{w} \in \R^{\abs{\mathcal{V}}} : w_{v} \geq 0 \text{ for all } v \in \mathcal{V} \text{ and } \sum \limits_{v \in \mathcal{V}} w_{v} = 1\right\},
\end{split}
\end{equation}
where, in general, $\R^d$ denotes the set of $d$-dimensional vectors with real-valued entries and $\abs{\cdot}$ of a set denotes its cardinality (i.e., the number of elements in that set).

These weights determine how much relative emphasis the analysis places on each validation period's Difference-in-Trends. The set restriction in \eqref{eq: max set iden} is a special case of the weighted model described in \eqref{eq: weighted set iden} and \eqref{eq: weights I}. In this special case, all weight is assigned to the pre-periods that attain the maximum absolute Difference-in-Trends, while every other validation period receives weight $0$.

\section{Motivating an Alternative to Pretrends-based Approaches} \label{sec: motivation}

To motivate an alternative to a pretrends-based sensitivity analysis, \Cref{fig: DID cases} below presents two stylized examples adapted from the two-period patterns in \citet{keeleetal2019} and \citet[][pp.~164--165]{rosenbaum2017} by adding a second pre-treatment period so that pretrends are exactly parallel.
\begin{figure}[!htb]
\includegraphics[width = \linewidth]{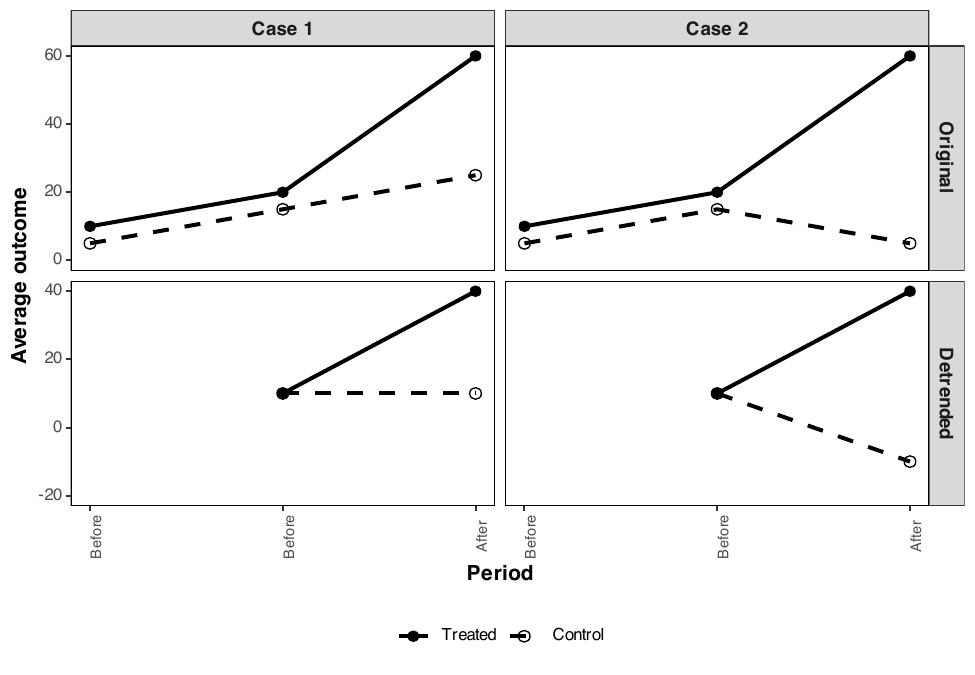}
\caption{Two stylized cases (columns), each showing trends for treated and control groups across three time periods: two pre-treatment periods and one post-treatment period. The $x$-axis is common to all panels and reflects the time dimension, while the $y$-axis represents the average of the outcome variable, with the scale of the outcome varying by row. In the top row, outcomes are plotted on their original scale; in the bottom row, outcomes are transformed into changes relative to the group's mean outcome in the immediately preceding period. 
}
\label{fig: DID cases}
\end{figure}
\noindent These two cases both have perfectly parallel pretrends. However, in Case 2, the post-period outcome in the control group exhibits a shock in which the trend breaks from its pre-period pattern. In Case 1, by contrast, no such shock exists, as the control group's pretrend continues uninterrupted into the post-period. A control group that stays on the group's pretrend strengthens the inference in the sense of \citet[][p.~211]{rosenbaum2001}, who observes that the usual difference-in-differences analysis does not register the added support that a control group's stability in the absence of treatment provides.

In Case 2, the control group's post-period deviation from the control group's pre-period trend is consistent with two possible scenarios:
\begin{itemize}
\item \textbf{Scenario 1 (common shock).} Both treated and control groups experience a shared shock in the post-period. In this scenario, the treated group's pre-period trend in average outcomes is a poor imputation for the treated group's counterfactual post-period average outcome under no treatment. Instead, the control group's change in average outcomes better captures how the treated group's average outcome would have evolved given that both groups experienced the same shock.
\item \textbf{Scenario 2 (control-specific shock).} Only the control group experiences a post-period shock. In this scenario, the control group's change in average outcomes is a poor imputation for the treated group's counterfactual average outcome since the control group experienced a shock that the treated group did not. Instead, the treated group's pre-period trend in average outcomes better captures how the treated group's average outcome would have evolved given that the treated group did not experience a shock.
\end{itemize}

In Case 2, the \textrm{ATT} based on the parallel trends imputation ($\tau_{0,T}$) --- appropriate under Scenario 1 (common shock) --- differs from that based on the treated group's detrended pre-period mean ($\tau_{1,T-1}$), which is more appropriate under Scenario 2 (control-specific shock). In other words, in Case 2, the \textrm{ATT} implied by the parallel trends imputation can be substantially in error when the parallel trends assumption is false and an alternative imputation is justified. This is not so in Case 1, where all such imputations yield the same \textrm{ATT}. This distinction captures the intuition that causal conclusions about the \textrm{ATT} depend more heavily on the parallel trends assumption in Case 2 than in Case 1.

This intuition is not reflected in a pretrends-based sensitivity analysis. Because such sensitivity analyses restrict $\delta_T$ using only the pre-period Differences-in-Trends, the control group's post-period change affects the location of the identified set (through the \textrm{DID}) but not the width of that set, so such analyses fail to distinguish between these two cases. Because Cases 1 and 2 have identical pretrends, both are, at the population level, perfectly robust --- under the sensitivity model in \eqref{eq: max set iden} --- to violations of parallel trends that are some magnitude, $M \geq 0$, of the greatest pre-period violation of parallel trends.

These two cases also suggest an important distinction between the plausibility of parallel trends and robustness to violations of parallel trends (due to concordance across multiple imputations). It is standard practice to interpret parallel pretrends as indicative of the plausibility of parallel trends. However, as a growing body of literature \citep{kahn-langlang2020,rothsantanna2023,egamiyamauchi2023} points out, pretrends are informative about post-period parallel trends only under additional assumptions that link the two, such as the extended parallel trends assumption of \citet{egamiyamauchi2023}. In the absence of such assumptions, the distinction between pretrends and concordance highlights a potential misinterpretation: Parallel pretrends, often taken as evidence for the plausibility of parallel trends, may instead reflect concordance --- that is, robustness of conclusions to certain violations of parallel trends. The distinction between plausibility and dependence also explains why discordance is deliberately not invariant to a common post-period shock. Starting from Case 1, adding the same shock to both groups' post-period potential outcomes, under treatment and under control alike, does not change the \textrm{ATT} or $\delta_T$, but does move the parallel trends imputation away from every pre-period imputation; the discordance-based bounds therefore widen, while the pretrends-based bounds do not change. Scenario 1 and Scenario 2 are observationally equivalent. Even with population-level data, that is, with the joint distribution of the observable outcomes known exactly, nothing distinguishes the two scenarios, so no amount of data can reveal which imputation is correct. The wider bounds therefore do not say that parallel trends has become less plausible. The wider bounds say that the conclusion now depends on which scenario holds, because the two imputations, and hence the two values of the \textrm{ATT} that the imputations imply, differ by the size of the shock.

\section{Discordance-based Sensitivity Analysis} \label{sec: discordance-based model}

I now propose an alternative, discordance-based sensitivity analysis that reflects the intuition for why Case 1 is more robust than Case 2. The central insight is that the assumptions underlying different imputations are not necessarily mutually exclusive. When these imputations are concordant, the validity of using the treated group's pre-period trend to impute the treated group's counterfactual does \textit{not} imply that the control group's outcome change is an invalid imputation. The assumptions supporting both imputations are mutually exclusive only when the imputations themselves are discordant. Therefore, to the extent that the conclusion about the \textrm{ATT} under the parallel trends imputation holds under multiple assumptions rather than just one, this conclusion depends less on the truth of the parallel trends assumption.

This general logic is closely related to that of bracketing relationships for DID, whether the bracket arises from two identifying assumptions applied to the same comparison \citep{angristpischke2008,guryan2001,dingli2019,demetrescuetal2025} or from two control groups whose untreated trends bracket the treated group's untreated trend \citep{hasegawaetal2019,yeetal2024}. Setting aside technical details, the general logic of a bracketing relationship is based on what the \textrm{ATT} would be under one assumption and what the \textrm{ATT} would be under another. In some cases, the \textrm{ATT} under one assumption is guaranteed to be larger than the \textrm{ATT} under another, thereby implying that, if one mistakenly adopts one assumption when the other is true, the \textrm{ATT} must be bounded between two values. The closer the \textrm{ATT}s under these two assumptions, the better.

This general logic remains important even without the bounding relationship. Whether or not the \textrm{ATT} under one assumption is guaranteed to be larger or smaller than under another, what matters is the degree of discordance. The closer the two \textrm{ATT}s are, the more robust the causal conclusion becomes under the parallel trends assumption.

To set up this framework, first define the discordance between the parallel trends imputation and any alternative imputation based on the average of detrended potential outcomes under no treatment as
\begin{align} \label{eq: concordance}
\lambda_{g,v} \coloneqq \abs{\bar{\epsilon}_{g,v}(0) - \bar{\epsilon}_{0,T}(0)}
\end{align}
for any $(g,v) \in \mathcal{G} \times \mathcal{V}$. Two features of this notion of discordance deserve emphasis. First, because discordance is defined in terms of average changes, concordance of changes neither implies nor is implied by similarity of average levels across groups or periods, the quantities restricted by the bounded time- and interstate-variation assumptions of \citet[][pp.~237--238]{manskipepper2018}. Second, like parallel trends itself \citep{rothsantanna2023}, discordance is defined through means and so is not invariant to nonlinear monotone transformations of the outcome: Concordance on one scale (e.g., levels) does not imply concordance on another (e.g., logs).

The discordance in \eqref{eq: concordance} provides a benchmark for bounding the error of the parallel trends imputation. The sensitivity model assumes that this error of the parallel trends imputation, $\abs{\bar{\epsilon}_{1,T}(0) - \bar{\epsilon}_{0,T}(0)}$, satisfies the restriction
\begin{align} \label{eq: pointwise concord mod}
\abs{\bar{\epsilon}_{1,T}(0) - \bar{\epsilon}_{0,T}(0)} \leq M \lambda_{g,v}.
\end{align}
The intuition behind \eqref{eq: pointwise concord mod} is straightforward. The true counterfactual is unobservable and, hence, so is the error of the parallel trends imputation, represented by the left-hand side of the inequality in \eqref{eq: pointwise concord mod}. However, the observable quantity, $\bar{\epsilon}_{g,v}(0)$, can serve as a proxy for the unobservable counterfactual. With this proxy, one can calibrate the error in the parallel trends imputation based on the observable discordance between $\bar{\epsilon}_{g,v}(0)$ and $\bar{\epsilon}_{0,T}(0)$. Specifically, the model assumes that the error of the parallel trends imputation is bounded by some magnitude, $M \geq 0$, of the discordance $\lambda_{g,v}$. While this bound treats $\bar{\epsilon}_{g,v}(0)$ as the proxy for the counterfactual, the logic could be reversed: In principle, one could instead treat $\bar{\epsilon}_{0,T}(0)$ as a proxy and use it to bound the error in using $\bar{\epsilon}_{g,v}(0)$ as a counterfactual imputation.

The pointwise bound in \eqref{eq: pointwise concord mod} relies on a single proxy. However, a more flexible approach aggregates information across all possible proxies. Specifically, one can form a convex combination of the pointwise bounds to obtain the following model:
\begin{align} \label{eq: sens model I}
\abs{\bar{\epsilon}_{1,T}(0) - \bar{\epsilon}_{0,T}(0)} \leq M \sum \limits_{(g,v) \in \mathcal{G} \times \mathcal{V}} w_{g,v} \lambda_{g,v},
\end{align}
where the weights $\bm{w} \coloneqq \left\{w_{g,v}\right\}_{(g,v) \in \mathcal{G} \times \mathcal{V}}$, which lie in the unit simplex over $\mathcal{G} \times \mathcal{V}$ denoted by
\begin{equation} \label{eq: unit simplex over G x V}
\begin{split}
\Delta_{\mathcal{G} \times \mathcal{V}} & \coloneqq \left\{\bm{w} \in \R^{\abs{\mathcal{G} \times \mathcal{V}}} : w_{g,v} \geq 0 \text{ for all } (g,v) \text{ and } \sum \limits_{(g,v) \in \mathcal{G} \times \mathcal{V}} w_{g,v} = 1\right\},
\end{split}
\end{equation}
represent the informativeness of each possible proxy for the counterfactual.

Proposition \ref{prop: tight sens bound} in the appendix shows that this upper bound in \eqref{eq: sens model I} is sharp for every choice of weights $\bm{w}$. That is, for every $\bm{w} \in \Delta_{\mathcal{G} \times \mathcal{V}}$ and every value of the observable quantities $\bar{\epsilon}_{0,T}(0)$ and $\bar{\epsilon}_{g,v}(0)$ for all $(g,v) \in \mathcal{G} \times \mathcal{V}$, there exists a value of the unobservable $\bar{\epsilon}_{1,T}(0)$ consistent with \eqref{eq: sens model I} at which the inequality holds with equality, and every value of the \textrm{ATT} in the interval $\left[\tau_{0,T} - M \sum_{(g,v) \in \mathcal{G} \times \mathcal{V}} w_{g,v} \lambda_{g,v}, \, \tau_{0,T} + M \sum_{(g,v) \in \mathcal{G} \times \mathcal{V}} w_{g,v} \lambda_{g,v}\right]$ that \eqref{eq: sens model I} implies is attained by some such value. Thus, no strictly smaller upper bound is compatible with the model.

This weighted formulation in \eqref{eq: sens model I} allows researchers to incorporate multiple proxies and to encode judgments about the relative informativeness of each one. For instance, if all proxies are equally informative, one could assign uniform weights,
\begin{align*}
w_{g,v} = \dfrac{1}{\abs{\mathcal{G} \times \mathcal{V}}} \text{ for all } (g,v) \in \mathcal{G} \times \mathcal{V},
\end{align*}
which leads to the average-bound model of
\begin{align} \label{eq: avg sens model I}
\abs{\bar{\epsilon}_{1,T}(0) - \bar{\epsilon}_{0,T}(0)} \leq M \frac{1}{\abs{\mathcal{G} \times \mathcal{V}}} \sum_{(g,v) \in \mathcal{G} \times \mathcal{V}} \lambda_{g,v}.
\end{align}
Other choices reflect different judgments. Time-decaying weights can give more importance to proxies from periods closer to treatment, while group-specific weights can emphasize proxies from certain groups (e.g., the treated group).

In general, however, the choice of weights is difficult to justify, as any specification ultimately encodes prior judgments that may vary across analysts. To guard against undue reliance on subjective weighting choices, I adopt a weighting scheme for the discordance-based model in \eqref{eq: sens model I}, analogous to the approach in \eqref{eq: max set iden} from \citet{rambachanroth2023}. This weighting scheme places all weight on whichever group-periods attain the largest discordance from the parallel trends imputation.

Formally, this weighting scheme, $\bm{w}^{\mathrm{max}}$, is 
\begin{equation} \label{eq: minimax weights I}
w^{\mathrm{max}}_{g,v} \coloneqq
\begin{cases}
1 / \abs{\mathcal{A}} & \text{if } (g,v) \in \mathcal{A}, \\
0 & \text{otherwise}.
\end{cases}
\end{equation}
The set $\mathcal{A}$ in \eqref{eq: minimax weights I}, formally defined as
\begin{align*}
\mathcal{A} \coloneqq \{(g, v) \in \mathcal{G} \times \mathcal{V}: \lambda_{g,v} = \max \limits_{(g^{\prime}, v^{\prime}) \in \mathcal{G} \times \mathcal{V}} \lambda_{g^{\prime}, v^{\prime}}\},
\end{align*}
consists of the group-period pairs that attain the maximum discordance. Although the weights in \eqref{eq: minimax weights I} are chosen to be uniform on $\mathcal{A}$, any probability distribution supported on this set would yield the same bound.

These weights in \eqref{eq: minimax weights I} produce the sensitivity model of
\begin{align} \label{eq: minimax regret sens model I}
\abs{\bar{\epsilon}_{1,T}(0) - \bar{\epsilon}_{0,T}(0)} \leq M \max_{(g,v) \in \mathcal{G} \times \mathcal{V}} \lambda_{g,v}.
\end{align}
Under this model, the error of the parallel trends imputation is bounded by a magnitude, $M \geq 0$, of the largest discordance across all groups and validation periods. This bound in \eqref{eq: minimax regret sens model I} is the central discordance-based sensitivity model proposed in this paper. Because every $\lambda_{g,v}$ is a function of observable group-period means, the model in \eqref{eq: minimax regret sens model I} can also be expressed within the general framework of \citet{rambachanroth2023}; what distinguishes the discordance-based model from their leading restrictions is the benchmark used to bound violations of parallel trends, not the framework.

Returning now to the two stylized examples in \Cref{fig: DID cases}, the value of the sensitivity model in \eqref{eq: minimax regret sens model I} is clear. \Cref{fig: sens cases} below shows that, while both models yield identical sensitivity in Case 1, they diverge substantially in Case 2.
\begin{figure}[!htb] 
\includegraphics[width = \linewidth]{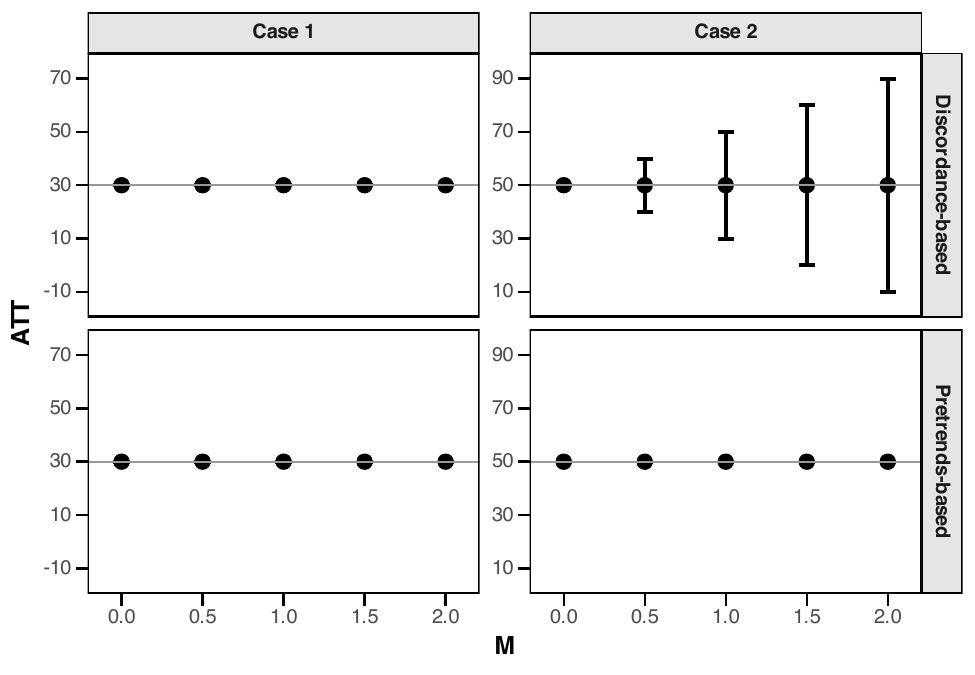}
\caption{Sensitivity bounds for the two stylized cases in \Cref{fig: DID cases}, with columns representing the cases and rows indicating the sensitivity model. The $x$-axis shows values of the sensitivity parameter $M \in \{0, 0.5, 1, 1.5, 2\}$, and the $y$-axis shows the corresponding values of the \textrm{ATT}. Solid points represent the \textrm{ATT} under the parallel trends assumption. Vertical bars indicate the identified set for the \textrm{ATT} under violations of the parallel trends assumption bounded by $M$ times the largest discordance (top row) or by $M$ times the largest pre-period Difference-in-Trends (bottom row); where no bar is visible, the identified set is a single point. The $y$-axis of each column is centered on that case's \textrm{ATT} under the parallel trends assumption and spans the same range in both columns, so that the widths of the vertical bars are comparable across columns.}
\label{fig: sens cases}
\end{figure}
\noindent Unlike the pretrends-based model in \eqref{eq: max set iden}, the discordance-based model in \eqref{eq: minimax regret sens model I} implies that the conclusion about the \textrm{ATT} depends more on parallel trends in Case 2 than in Case 1.

The discordance-based model in \eqref{eq: minimax regret sens model I} consists of two crucial components: a conception of discordance in \eqref{eq: concordance} and a maximum-discordance weighting scheme in \eqref{eq: minimax weights I}. I first justify the use of discordance as a principled foundation for the credibility of causal conclusions drawn under the parallel trends assumption, showing that agreement across multiple counterfactual imputations strengthens such conclusions. In doing so, this justification helps clarify two concepts often conflated in sensitivity analysis --- (1) the plausibility of an identifying assumption and (2) the extent to which causal conclusions depend on it --- and shows that concordance implies parallel pretrends, though not vice versa. I then provide a decision-theoretic justification demonstrating that the weighting scheme concentrating all mass on the maximal discordance is the minimax solution to a regret-based optimization problem. Together, these justifications provide a stronger foundation for the overall bound in \eqref{eq: minimax regret sens model I}.

\subsection{Justification for Discordance} \label{ssec: concordance}

As noted above, one can use $\bar{\epsilon}_{g,v}(0)$ as a proxy for the counterfactual $\bar{\epsilon}_{1,T}(0)$, or vice versa, to bound the error of the counterfactual imputation. Following logic closely related to bracketing relationships, consider two scenarios:
\begin{enumerate}
\item The parallel trends imputation, $\bar{\epsilon}_{0,T}(0)$, is correct, but the proxy, $\bar{\epsilon}_{g,v}(0)$, is not.
\item The proxy, $\bar{\epsilon}_{g,v}(0)$, is correct, but the parallel trends imputation, $\bar{\epsilon}_{0,T}(0)$, is not.
\end{enumerate}
In the first scenario, the parallel trends assumption holds exactly, and there is no error in the imputation. However, if one mistakenly assumes the first scenario when the second is actually true, then the size of the unobservable error will depend on the observable discordance between $\bar{\epsilon}_{0,T}(0)$ and $\bar{\epsilon}_{g,v}(0)$. Insofar as the second scenario could hold, then, all else equal, greater concordance between $\bar{\epsilon}_{0,T}(0)$ and $\bar{\epsilon}_{g,v}(0)$ implies a smaller overall error of the parallel trends imputation.

To formalize this reasoning within a sensitivity analysis framework based on set (rather than point) identification, consider the bound in \eqref{eq: pointwise concord mod}, $\abs{\bar{\epsilon}_{1,T}(0) - \bar{\epsilon}_{0,T}(0)} \leq M \lambda_{g,v}$. This bound is the set identification analogue of the point identification condition in the first scenario, whereby the parallel trends imputation is correct ($\bar{\epsilon}_{1,T}(0) = \bar{\epsilon}_{0,T}(0)$) and the proxy is not ($\bar{\epsilon}_{1,T}(0) \neq \bar{\epsilon}_{g,v}(0)$). Specifically, when $M = 0$ and $\lambda_{g,v} > 0$, the bound in \eqref{eq: pointwise concord mod} collapses to the point-identifying condition of the first scenario.

The second scenario corresponds to an alternative bound,
\begin{align} \label{eq: pointwise concord mod rev}
\abs{\bar{\epsilon}_{1,T}(0) - \bar{\epsilon}_{g,v}(0)} \leq M \lambda_{g,v}.
\end{align}
This bound in \eqref{eq: pointwise concord mod rev} is the set identification analogue of the point identification condition in which the proxy is correct ($\bar{\epsilon}_{1,T}(0) = \bar{\epsilon}_{g,v}(0)$) and the parallel trends imputation is not ($\bar{\epsilon}_{1,T}(0) \neq \bar{\epsilon}_{0,T}(0)$). Specifically, when $M = 0$ and $\lambda_{g,v} > 0$, the bound in \eqref{eq: pointwise concord mod rev} collapses to the point-identifying condition of the second scenario.

In \eqref{eq: pointwise concord mod rev}, $\bar{\epsilon}_{g,v}(0)$ serves as the imputation and $\bar{\epsilon}_{0,T}(0)$ as the proxy. One could therefore redefine discordance as the distance between the imputation $\bar{\epsilon}_{g,v}(0)$ and the detrended mean of any other group-period index chosen to play the role of the proxy, not only $(0,T)$. Introducing this additional notation, however, would add complexity without changing the analysis: Because the proxy's group-period index in \eqref{eq: pointwise concord mod rev} is specifically $(0,T)$, and discordance in \eqref{eq: concordance} is defined using an absolute value, the discordance in \eqref{eq: pointwise concord mod rev} is identical to that in \eqref{eq: pointwise concord mod}.

Building on this alternative bound in \eqref{eq: pointwise concord mod rev}, recall that $\tau_{0,T}$ and $\tau_{g,v}$ refer to the imputed values of the \textrm{ATT} when one imputes the counterfactual using $\bar{\epsilon}_{0,T}(0)$ or $\bar{\epsilon}_{g,v}(0)$, respectively. When using the set restriction in \eqref{eq: pointwise concord mod}, analogously denote the resulting lower and upper bounds of the imputed \textrm{ATT} by $\underline{\tau}_{0,T}$ and $\overline{\tau}_{0,T}$. By contrast, if one instead uses the set restriction in \eqref{eq: pointwise concord mod rev}, I denote the corresponding bounds by $\underline{\tau}_{g,v}$ and $\overline{\tau}_{g,v}$.

Now suppose a prior distribution, $\pi$, over whether these two restrictions in \eqref{eq: pointwise concord mod} and \eqref{eq: pointwise concord mod rev} hold. The corresponding probabilities are defined as follows:
\begin{itemize}
\item $p_{(0,T),(g,v)}$ denotes the probability that both restrictions hold;
\item $p_{(0,T),\neg(g,v)}$ and $p_{\neg(0,T),(g,v)}$ denote the respective probabilities that only one holds; and
\item $p_{\neg(0,T),\neg(g,v)}$ denotes the probability that neither holds.
\end{itemize}
\Cref{prop: concordance} below shows that, so long as $p_{\neg(0,T),(g,v)} > 0$, the combined expected error of the \textrm{ATT}'s bounds
under the set restriction in \eqref{eq: pointwise concord mod} is increasing
in discordance.
\begin{proposition} \label{prop: concordance}
Let $\overline{\tau}$ and $\underline{\tau}$ denote the upper and lower bounds of the \textrm{ATT} under whichever of the restrictions in \eqref{eq: pointwise concord mod} and \eqref{eq: pointwise concord mod rev} actually hold, and suppose that the errors of these bounds in the state in which neither restriction holds are fixed but unknown quantities. If $p_{\neg(0,T),(g,v)} > 0$, then, for any fixed $M \geq 0$, the combined expected error of the two bounds, $\E_{\pi}\left[\abs{\overline{\tau} - \overline{\tau}_{0,T}}\right] + \E_{\pi}\left[\abs{\underline{\tau} - \underline{\tau}_{0,T}}\right]$, is strictly increasing in $\lambda_{g,v}$, where $\E_{\pi}\left[\cdot\right]$ denotes the expectation with respect to the joint prior distribution, $\pi$.
\end{proposition}
\noindent \Cref{prop: concordance} directly ties discordance to the robustness of causal conclusions.

\Cref{prop: concordance} also clarifies the distinction between the plausibility of the parallel trends identification condition --- formally, its marginal probability $p_{(0,T),(g,v)} + p_{(0,T),\neg(g,v)}$ --- and the extent to which causal conclusions depend on that condition. To explicate this distinction between dependence on and plausibility of the parallel trends assumption, recall the two set restrictions, one for the parallel trends imputation, $\bar{\epsilon}_{0,T}(0)$, and one for the alternative imputation, $\bar{\epsilon}_{g,v}(0)$:
\begin{itemize}
\item \Cref{eq: pointwise concord mod}, namely $\abs{\bar{\epsilon}_{1,T}(0) - \bar{\epsilon}_{0,T}(0)} \leq M \lambda_{g,v}$, 
\item \Cref{eq: pointwise concord mod rev}, namely $\abs{\bar{\epsilon}_{1,T}(0) - \bar{\epsilon}_{g,v}(0)} \leq M \lambda_{g,v}$.
\end{itemize}
If both restrictions are false, then the error of the bounds on the \textrm{ATT} under \eqref{eq: pointwise concord mod} is some fixed but unknown value. However, if \eqref{eq: pointwise concord mod rev} is true and \eqref{eq: pointwise concord mod} is false, then the error of the \textrm{ATT}'s bounds under the incorrect restriction is simply the difference between the two bounding sets. This difference is driven entirely by the discordance between the parallel trends imputation, $\bar{\epsilon}_{0,T}(0)$, and the alternative imputation, $\bar{\epsilon}_{g,v}(0)$. Thus, discordance is not directly about plausibility. Rather, discordance captures the degree of dependence of inferences on the parallel trends identification condition; that is, if the condition were violated, how different would the conclusions be?

That said, concordance is informative only insofar as the identifying condition underlying the imputed quantity $\bar{\epsilon}_{g,v}(0)$ is itself at least somewhat plausible. A conclusion that would change little under a different identifying assumption is not meaningful if that alternative assumption is wholly implausible. This is why \Cref{prop: concordance} includes in its if–then statement the antecedent condition $p_{\neg(0,T),(g,v)} > 0$, reflecting a positive probability that the parallel trends assumption is false and that the alternative assumption holds. 

\Cref{prop: concordance} offers a formal justification for why one ought to quantify robustness to violations of parallel trends in terms of discordance between the parallel trends imputation and alternative imputations. Importantly, a discordance-based model accords with much of the intuition underlying parallel pretrends. I now show how concordance and parallel pretrends are connected: Concordance implies parallel pretrends, but the converse is not true.

To formalize this argument, consider the detrended means of each group ($G = 1$ and $G = 0$) for any validation period in $\mathcal{V}$. The joint discordance is defined as a weighted sum of the absolute discordances of these detrended means from the benchmark $\bar{\epsilon}_{0,T}(0)$. \Cref{prop: concordance parallel trends} shows that, for any $v \in \mathcal{V}$, a small upper bound on joint discordance is sufficient to ensure a small violation of parallel pretrends.
\begin{proposition} \label{prop: concordance parallel trends}
For any pre-treatment validation period, $v \in \mathcal{V}$, the violation of parallel trends, $\abs{\bar{\epsilon}_{1,v}(0) - \bar{\epsilon}_{0,v}(0)}$, satisfies
\begin{align} \label{eq: pretrends bound}
\abs{\bar{\epsilon}_{1,v}(0) - \bar{\epsilon}_{0,v}(0)} \leq 2 \max\left\{\lambda_{1,v}, \, \lambda_{0,v} \right\},
\end{align}
where the maximum on the right-hand side is an upper bound on the joint discordance between $\bar{\epsilon}_{0,T}(0)$ and the pair $(\bar{\epsilon}_{1,v}(0), \, \bar{\epsilon}_{0,v}(0))$.
\end{proposition}
\noindent The proposition establishes that when the joint discordance is constrained to be small (i.e., its upper bound is small), the violation of parallel pretrends must also be small. In the limiting case in which the joint discordance is constrained to be zero, parallel pretrends cannot be violated. However, the converse does not hold. A small violation of parallel pretrends does not imply that the upper bound on discordance is small.


\subsection{Justification for Maximum-Discordance Weighting} \label{ssec: weighting}

As noted above, the discordance-based sensitivity model in \eqref{eq: minimax regret sens model I} has two key components: a conception of discordance in \eqref{eq: concordance} and a maximum-discordance weighting scheme in \eqref{eq: minimax weights I}. Thus far, I have provided a foundation justifying why discordance offers a reasonable basis for restricting violations of the parallel trends assumption. I now turn to justifying the specific weights in \eqref{eq: minimax regret sens model I}, which concentrate all mass on the $(g,v) \in \mathcal{G} \times \mathcal{V}$ that attain the maximum discordance.

The maximum-discordance weighting scheme has an important decision-theoretic justification as the solution to an adversarial game. In this game, both the worst-case violation of parallel trends and the bound implied by any proposed weighting scheme are determined under the same sensitivity model for a given $M \geq 0$. One then selects weights to minimize worst-case regret, defined as the largest possible gap between the worst-case violation of parallel trends and the bound implied by the selected weighting scheme. The weights in \eqref{eq: minimax weights I} are those that solve this minimax regret problem.

As others have noted in different sensitivity analysis frameworks \citep[e.g.,][]{cohenetal2020}, it is intuitive to think of an adversarial game as a way to justify using the worst-case (widest) set restriction for any given $M \geq 0$. If the adversary represents future counterclaims against a study's causal conclusion, then adopting the widest bounds ensures that no counterclaim can overturn the researcher's conclusion for that value of $M \geq 0$. If a researcher can rule out certain values of the \textrm{ATT} under the worst-case (widest) sensitivity interval, then those values must also be ruled out under any narrower intervals.

To formally establish that the weights in \eqref{eq: minimax weights I} minimize worst-case regret, let $\psi \coloneqq \abs{\bar{\epsilon}_{1,T}(0) - \bar{\epsilon}_{0,T}(0)}$ denote the absolute error of the parallel trends imputation. Under the sensitivity model in \eqref{eq: sens model I}, $\psi$ is only partially identified: For all possible weights $\bm{w} \in \Delta_{\mathcal{G}\times\mathcal{V}}$, where $\Delta_{\mathcal{G}\times\mathcal{V}}$ denotes the unit simplex over $\mathcal{G} \times \mathcal{V}$ defined in \eqref{eq: unit simplex over G x V}, the model in \eqref{eq: sens model I} implies
\begin{align*}
0 \leq \psi \leq M \sum_{(g,v) \in \mathcal{G} \times \mathcal{V}} w_{g,v}\lambda_{g,v}.
\end{align*}
Thus, because the weights are not known, the set of all values of $\psi$ consistent with the sensitivity model for some admissible choice of weights is
\begin{align} \label{eq: Psi}
\Psi & \coloneqq \left\{\psi \in \mathbb{R}_{\geq 0}: \psi \leq M \sum_{(g,v) \in \mathcal{G} \times \mathcal{V}} w_{g,v} \lambda_{g,v} \text{ for some } \bm{w} \in \Delta_{\mathcal{G}\times\mathcal{V}} \right\},
\end{align}
where $\mathbb{R}_{\geq 0}$ is the set of nonnegative reals.

For any $\bm{w} \in \Delta_{\mathcal{G} \times \mathcal{V}}$, the corresponding worst-case regret is how much larger the absolute error of the parallel trends imputation could be --- given that $\psi$ must lie in $\Psi$ --- relative to the bound implied by that weighting scheme. Formally, this worst-case regret is
\begin{align} \label{eq: worst-case regret I}
R(\bm{w}) & \coloneqq \sup_{\psi\in\Psi} \left\{\psi - M \sum_{(g,v) \in \mathcal{G} \times \mathcal{V}} w_{g,v} \lambda_{g,v} \right\}.
\end{align}
The $\sup_{\psi\in\Psi}\{\cdot\}$ in \eqref{eq: worst-case regret I} denotes the least upper bound of the set of all values that the expression in braces, $\{\cdot\}$, attains as $\psi$ ranges over $\Psi$.

\Cref{prop: equivalence} now establishes that the model in \eqref{eq: minimax regret sens model I}, which places all weight on the group-periods that attain the maximum discordance, is equivalent to the model in \eqref{eq: sens model I} with weights that minimize this worst-case regret in \eqref{eq: worst-case regret I}.
\begin{proposition}
\label{prop: equivalence}
The sensitivity model in \eqref{eq: minimax regret sens model I}, which uses the weighting scheme $\bm{w}^{\max}$ in \eqref{eq: minimax weights I}, is equivalent to the general model in \eqref{eq: sens model I} with weights that solve the following minimax regret optimization problem:
\begin{align} \label{eq: minimax sol}
\bm{w} \in \argmin \limits_{\bm{w}^{\prime} \in \Delta_{\mathcal{G} \times \mathcal{V}}} R(\bm{w}^{\prime}).
\end{align}
\end{proposition}
\noindent The optimization problem in \eqref{eq: minimax sol}, which is expressed in terms of the error of the parallel trends imputation for a given weighting scheme, implicitly targets the width of the \textrm{ATT}'s bounds implied by that scheme, effectively selecting the weighting scheme that yields the widest identified set for the \textrm{ATT} permitted by the model. Because the regret in \eqref{eq: worst-case regret I} penalizes only the possibility that the true error of the parallel trends imputation exceeds the reported bound and attaches no cost to the width of that bound, the minimax solution is necessarily the widest admissible bound. \Cref{prop: equivalence} should therefore be read as a formal statement that no alternative admissible weighting can overturn the maximum-discordance weighting (up to ties on $\mathcal{A}$), rather than as the solution to a problem that balances the width of the bounds on the \textrm{ATT} against the guarantee that the true error of the parallel trends imputation stays within those bounds, as in confidence intervals for partially identified parameters \citep{imbensmanski2004} and the fixed-length confidence intervals of \citet[][Section~4.3]{rambachanroth2023}.

Crucially, for any $M \geq 0$, the identified intervals for all $(g,v) \in \mathcal{G} \times \mathcal{V}$ share the same center --- the \textrm{ATT} implied by the parallel trends imputation. Different choices of weights, $\bm{w}$, affect only the width of the bounds, not their location. The minimax weights correspond to the maximum possible width across all feasible weighting schemes. Thus, any interval constructed using weights other than the worst-case weights must be bracketed within the worst-case interval. Consequently, no alternative weighting scheme can imply a value of the \textrm{ATT} that falls outside the researcher's bounds, and any such alternative necessarily implies less sensitivity than the minimax choice.

\section{Extension to Multiple Post-Periods}  \label{sec: multiple post periods}

Up to this point, the discordance-based sensitivity model has focused on the simple setting of a single post-treatment period. To accommodate settings with multiple post-treatment periods, it is helpful to begin by clarifying the causal estimands of interest. With several post-periods, the treated group $G = 1$ now has an \textrm{ATT} for each post-treatment period. I first consider the case in which the target is an aggregated estimand --- namely, a linear combination of these period-specific \textrm{ATT}s --- which requires only a minor modification of the notation introduced above. I then consider the case in which each post-period \textrm{ATT} is itself an estimand of interest, a setting that necessitates additional notation.

Let $T_0 \geq 2$ denote the number of pre-treatment periods and let $T_1 \geq 1$ denote the number of post-treatment periods, so that $T_0 + T_1 = T$. Define the sets of pre- and post-treatment periods as $\mathcal{T}_{\mathrm{pre}} \coloneqq \{1, \ldots, T_0\}$ and $\mathcal{T}_{\mathrm{post}} \coloneqq \{T_0 + 1, \ldots, T\}$, respectively. In the earlier setup with a single post-period, $T_1 = 1$ and $\mathcal{T}_{\mathrm{post}} = \{T\}$.

The treated group $G = 1$ first enters treatment in period $T_0 + 1$. For each post-treatment period $s \in \mathcal{T}_{\mathrm{post}}$, define the period-$s$ \textrm{ATT} as
\begin{align}
\tau^s \coloneqq \bar{Y}_{1,s}(1) - \bar{Y}_{1,s}(0).
\end{align}
With multiple post-treatment periods, the treated group $G = 1$ therefore has a distinct \textrm{ATT} for each $s \in \mathcal{T}_{\mathrm{post}}$.

Denote the vector of post-treatment \textrm{ATT}s for group $G = 1$, with one component for each post-treatment period, by
\begin{align}
\bm{\tau} \coloneqq \left( \tau^{T_0 + 1},\, \ldots,\, \tau^{T} \right)^{\top},
\end{align}
where the superscript $^{\top}$ denotes transposition. Following \citet[][p.~2561]{rambachanroth2023}, define an aggregated estimand as $\theta \coloneqq \bm{\ell}^{\top}\bm{\tau}$, where $\bm{\ell}$ is a known vector of the same dimension as $\bm{\tau}$ whose elements sum to one. This formulation accommodates simple averages and weighted averages of the post-treatment \textrm{ATT}s.

When the estimand of interest is this aggregated quantity $\theta$, one can
replace the post-period quantity $Y_{g,T}(z)$ from the earlier framework with $\bm{\ell}^{\top}\bm{Y}_{g,\mathrm{post}}(z)$, where $\bm{Y}_{g,\mathrm{post}}(z)$ collects the treated ($z = 1$) or control ($z = 0$) potential outcomes for group $g$ across all
post-treatment periods $s \in \mathcal{T}_{\mathrm{post}}$, detrended relative to the last pre-treatment mean $\bar{Y}_{g,T_0}$ as in \eqref{eq: detrended outcome IV} below. Because the elements of $\bm{\ell}$ sum to one, all algebraic steps and arguments from the single-period setting apply directly. The sensitivity analysis then evaluates robustness for the aggregated estimand $\theta$.

By contrast, when the analysis targets the \textrm{ATT}s for each post-treatment period individually, additional notation is in order. Additional notation is needed because outcomes should be detrended via a group's mean in pre-treatment periods so that one has an uncontaminated prediction of the treated group's mean counterfactual outcome in the post-period. Using a period already affected by treatment would presumably yield a distorted prediction of the average outcome that would have arisen under no treatment. For the control group, which is not treated in the post-period, the same reference period is used to ensure symmetry, so that differences between the treated group's counterfactual detrended mean and the control group's observable detrended mean reflect deviations from parallel trends rather than artifacts of differing baseline periods.

Therefore, redefine the detrended outcome for group $g$ in period $t$ using the most recent pre-treatment period, which, for any period $t \in \{2,\ldots,T\}$, is $\min\{t - 1, T_0\}$. Formally, this detrended outcome is
\begin{align} \label{eq: detrended outcome IV}
\epsilon_{g,t} \coloneqq Y_{g,t} - \bar{Y}_{g,\min\{t - 1, T_0\}}.
\end{align}
When $t$ is itself a pre-treatment period, \eqref{eq: detrended outcome IV} reduces to detrending the outcome using the group's mean in the immediately preceding period, just as in the single post-period setup. However, when $t$ is a post-treatment period, the outcome is detrended using the most recent pre-treatment period, ensuring that the detrending baseline is not contaminated by treatment.

As before, detrending preserves the \textrm{ATT}. The \textrm{ATT} expressed in
terms of these detrended outcomes in \eqref{eq: detrended outcome IV} is algebraically equivalent to the \textrm{ATT} defined using the raw potential outcomes. The difference now is that the treated group $G = 1$ has a distinct \textrm{ATT} for each post-treatment period, $s \in \mathcal{T}_{\mathrm{post}}$.

For any such $s$, I analogously define a sensitivity model in which the
violation of parallel trends is governed by some magnitude $M \geq 0$ of the
maximum discordance between the parallel trends imputation and any other
imputation based on pre-treatment validation periods. Let
$\mathcal{V} \subseteq \mathcal{T}_{\mathrm{pre}} \setminus \{1\} = \{2, \ldots, T_0\}$ denote the set of
pre-treatment validation periods selected for analysis, with each
$v \in \mathcal{V}$ indexing a single pre-treatment period. Denote the
discordance between the parallel trends imputation and the average detrended
outcome based on group $g$ and validation period $v$ as $\lambda^{s}_{g,v} \coloneqq \abs{\bar{\epsilon}_{g,v}(0) - \bar{\epsilon}_{0,s}(0)}$. Then an immediate extension of the discordance-based sensitivity model in \eqref{eq: minimax regret sens model I} is
\begin{align}
\label{eq: minimax regret sens model mult post}
\abs{\bar{\epsilon}_{1,s}(0) - \bar{\epsilon}_{0,s}(0)} \leq M \max_{(g,v) \in \mathcal{G} \times \mathcal{V}} \lambda^{s}_{g,v}
\end{align}
for all $s \in \mathcal{T}_{\mathrm{post}}$.

Under this model in \eqref{eq: minimax regret sens model mult post}, the sensitivity analysis applies period-by-period in exactly the same manner as in the single post-period case. Various ``shock'' scenarios --- such as a control-specific shock, a common shock, or more complex temporal patterns --- remain fully encompassed by this framework. These scenarios merely alter the observed values of the discordance terms $\lambda^{s}_{g,v}$, which in turn determine the width of the sensitivity bounds on $\tau^s$. Thus, even with multiple post-treatment periods, the core principle persists: Sensitivity is governed by the extent to which the parallel trends imputation and alternative imputations based on pre-treatment periods disagree. One caveat applies: Because $\bar{\epsilon}_{0,s}(0)$ in \eqref{eq: detrended outcome IV} is a change over $s - T_0$ periods whereas each $\bar{\epsilon}_{g,v}(0)$ is a one-period change, $\lambda^{s}_{g,v}$ can grow with the horizon $s - T_0$ even under parallel trends, because $\bar{\epsilon}_{0,s}(0)$ accumulates the control group's period-to-period changes. If, for example, both groups' average untreated outcomes change by the same amount $b$ every period, then $\bar{\epsilon}_{0,s}(0) = (s - T_0) b$ whereas each $\bar{\epsilon}_{g,v}(0) = b$, so $\lambda^{s}_{g,v} = (s - T_0 - 1) \abs{b}$ although parallel trends holds exactly. In applications with long horizons or unequally spaced periods, one may therefore wish to express all detrended outcomes per unit of elapsed time before computing discordance, in the spirit of the consecutive-period benchmark in \citet[][fn.~7]{rambachanroth2023}.

The distinction between targeting the aggregated estimand $\theta$ and the period-specific \textrm{ATT}s is important. Because the maximization in \eqref{eq: minimax regret sens model mult post} is nonlinear, a sensitivity analysis applied separately to each post-period and then aggregated using $\bm{\ell}$ is generally not equivalent to one applied directly to $\theta$. Therefore, when the parameter of interest is the linear combination $\theta$, the sensitivity analysis should be applied directly to $\theta$, as the worst-case discordance framework is designed to provide a single robustness assessment for a single estimand --- whether aggregated or not --- rather than an aggregation of robustness assessments computed separately for multiple estimands.

\section{Application} \label{sec: application}

I now present an applied example that illustrates how the discordance-based sensitivity analysis can yield conclusions different from those of existing approaches. The example is drawn from \citet{wilse-samson2013}, which investigates whether the sudden withdrawal of foreign mine labor from Malawi and Mozambique in the mid-1970s affected electoral support for apartheid policies in South African districts that depended on that labor. The aim of the analysis is to estimate the average effect of this labor supply shock on the voting behavior of White South Africans in mining districts. This application fits squarely within the canonical DID framework, featuring a well-defined treatment shock, variation in exposure across districts, and observations from both before and after treatment. Related studies have also leveraged this labor supply shock in their own DID designs \citep[e.g.,][]{dinkelmanmariotti2016,dinkelmanetal2024}.

As \citet{wilse-samson2013} explains --- and as comprehensive historical accounts detail \citep{dekiewiet1941, thompson1990, welsh2009, dubow2014, omeara1996, feinstein2005} --- South Africa's apartheid regime faced a structural tension: On the one hand, the regime required large, concentrated labor forces to power urban industries such as mining. On the other hand, the apartheid regime sought to limit Black urbanization out of fear of Black majority rule, instead relegating Black South Africans to fragmented ``homelands.'' Presumably in an effort to manage this tension, the apartheid regime relied on foreign migrant labor to work in the mines, a strategy that many historians interpret as an attempt to reduce the political risks posed by a large, urbanized Black domestic workforce.

This balance was disrupted by a dramatic labor shock between South Africa's 1974 and 1977 national elections \citep{crush1993}. In April 1974, a plane carrying Malawian migrant laborers home from South Africa crashed, killing 74 of them. In response, Malawian President Hastings Banda suspended all migrant labor to South Africa. Around the same time, Portugal's \textit{Revolu\c{c}\~ao dos Cravos} in 1974 and Mozambique's subsequent independence in 1975 sharply curtailed the flow of Mozambican labor. These events --- external to South Africa's mining sector --- reduced the number of foreign mine workers from 336{,}000 at the end of 1973 to 208{,}000 at the end of 1977 \citep{wilse-samson2013}.

Likely in response to pressure from the mining industry, the apartheid regime allowed companies to expand recruitment of domestic Black South Africans. Their numbers rose between the end of 1973 and the end of 1977 from 86{,}000 to 214{,}000, often under longer contracts and with higher wages, contributing to a more stable urban workforce. As a result, the foreign share of the Black workforce on the gold mines fell from 78\% in 1974 to 47\% at the end of 1977 \citep[][pp.~8--9]{wilse-samson2013}. Given the apartheid regime's aforementioned structural tension, this shift may have provoked fears of Black majority rule and, in turn, heightened support for far-right parties seeking to entrench apartheid. A DID design based on the data in \citet{wilse-samson2013} offers a way to investigate this possibility.

The data from \citet{wilse-samson2013} include $153$ electoral districts with right-wing vote share measured in elections both before and after the labor shock. Of these $153$ districts, $14$ contained at least one gold mine (and thus were directly exposed to the labor shock), while the remaining $139$ did not. The pre-period includes the 1961, 1966, 1970, and 1974 national parliamentary elections. The post-period is collapsed into a single time point, denoted as 1977+. The outcome --- right-wing vote share --- is the proportion of voters in a district who supported parties to the right of the ruling National Party (NP), such as the Herstigte Nasionale Party (HNP) and, after its formation in 1982, the Conservative Party (CP).

\Cref{fig: wilse-samson I} below shows the trends in right-wing vote share for districts with and without mines before and after the labor shock.
\begin{figure}[!htb]
\includegraphics[width = \linewidth]{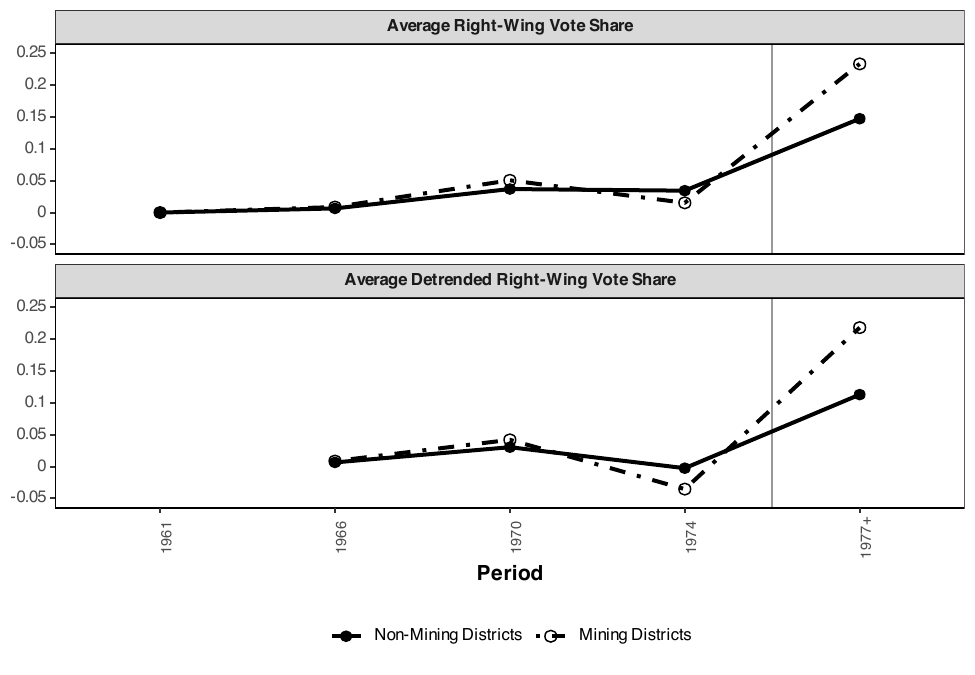}
\caption{The average right-wing vote share in electoral districts with at least one gold mine (dashed line) and those without any gold mines (solid line). The top panel plots the average right-wing vote share on its original scale, while the bottom panel plots the detrended series, defined as each group's change in average right-wing vote share relative to the group's mean in the preceding election period. The $x$-axis shows national election years. The election years 1961, 1966, 1970, and 1974 are shown as distinct time points, while ``1977+'' represents the post-treatment election years (1977, 1981, and 1987). In the top panel, each point reflects the average right-wing vote share across all mining or non-mining districts for that election year. For ``1977+'', the value is computed by first averaging within each district across all post-treatment election years and then averaging these district-level means within each group. In the bottom panel, each point reflects the average detrended outcome, with detrending for ``1977+'' performed after first computing each district's average outcome across all post-treatment periods. The vertical gray line indicates the start of the treatment period for mining districts, beginning with the first election after 1974.}
\label{fig: wilse-samson I}
\end{figure}
\noindent As \Cref{fig: wilse-samson I} shows, trends in right-wing vote share are relatively stable across mining and non-mining districts before 1974. After 1974, support increases in both groups, but the rise is substantially larger in mining districts.

Under the point identification condition that $M = 0$ in \eqref{eq: minimax regret sens model I}, which implies parallel trends, the estimated \textrm{ATT} is $0.11$. Using a cluster bootstrap at the electoral district level, the estimated standard error (SE) is $0.02$. The corresponding 95\% bootstrap percentile confidence interval (CI) is $[0.06, 0.15]$. Because the bootstrap resamples whole districts, each draw preserves each district's observed pre- and post-treatment outcomes, and the detrended outcomes are recomputed from the group means of the resampled data.

The control group's large post-period deviation from the pre-period trend suggests that the estimated \textrm{ATT} is highly dependent on the parallel trends assumption. In other words, suppose the average change from 1974 to 1977+ in the non-mining districts did \textit{not} represent what the average change would have been in the mining districts without the labor shock. If instead the pre-period change in mining districts from 1970 to 1974 were the correct imputation for the counterfactual, then the \textrm{ATT} would be dramatically different ($0.25$ compared to $0.11$).

Put differently, the discordance is large. To see this directly from \Cref{fig: wilse-samson I}, note first that the parallel trends imputation is $0.113$, the average detrended post-treatment outcome for the control group (non-mining districts). When combined with the treated group's average detrended post-treatment outcome of $0.218$, the resulting \textrm{ATT} point estimate is approximately $0.11$. As can be seen directly from \Cref{fig: wilse-samson I}, the remaining possible counterfactual imputations --- ordered chronologically from 1966 through 1974 and, within each election year, by group (0, then 1) --- are $0.007$, $0.009$, $0.030$, $0.042$, $-0.003$, and $-0.035$. Thus, the maximum discordance between the parallel trends imputation ($0.113$) and any other feasible imputation is approximately $\abs{0.113 - (-0.035)} \approx 0.148$. Up to rounding, this $0.148$ is exactly the difference between the two imputed \textrm{ATT}s ($0.25$ and $0.11$) described in the paragraph above: the former obtained using the treated group's most recent pre-period average detrended outcome and the latter using the control group's post-treatment average detrended outcome.

This maximum discordance of approximately $0.15$ implies that the \textrm{ATT}'s estimated bounds for, say, $M = 1$ are roughly $[0.11 - 0.15, \, 0.11 + 0.15] \approx [-0.04, \, 0.25]$. Accounting for sampling uncertainty is more challenging when $M > 0$ than when $M = 0$. Because the sample analogue of \eqref{eq: minimax regret sens model I} includes a maximum operator, the usual bootstrap --- appropriate in the $M = 0$ case above --- may be inconsistent, such as when the population-level joint CDF assigns the maximal discordance to multiple group-periods. 

I therefore draw on the ``intersection--union'' approach of \citet{bergerhsu1996}, which is known to produce valid, though often conservative, confidence intervals \citep[see][pp.~2262--2264, who document the conservativeness of this approach for bounds of the same form and propose a less conservative bootstrap]{yeetal2024}. In practice, I use the aforementioned district-level bootstrap procedure in which, for each bootstrap replication, I resample electoral districts with replacement and recompute both the DID estimator and the discordance for every group-period pair in $\mathcal{G} \times \mathcal{V}$ on the resulting resample. This yields a bootstrap distribution of the DID estimator and, for each group-period pair, a bootstrap distribution of discordance. For each pair, I then combine the bootstrap draws of the DID estimator and that pair's discordance to obtain bootstrap distributions for the corresponding lower and upper ATT bounds under the sensitivity model, from which I take the $\alpha/2$ and $1-\alpha/2$ quantiles. Finally, following the intersection–union logic of \citet{bergerhsu1996}, I combine these candidate intervals across all group-period pairs by taking the smallest lower endpoint and the largest upper endpoint. The resulting interval is the union of the intervals computed for the individual group-period pairs; \citet[][pp.~2262--2263]{yeetal2024} describe the same construction under the same name. The interval is valid, though conservative, because the interval computed for the pair that attains the maximal discordance in the population contains each endpoint of the identified set with probability at least $1 - \alpha/2$, and the union contains that interval; the union therefore covers the whole identified set for the \textrm{ATT}, and hence the \textrm{ATT} itself, with probability at least $1 - \alpha$, whichever pair attains the maximal discordance.

\Cref{fig: wilse-samson sens} shows that the \textrm{ATT}'s estimated bounds under the discordance-based model are highly sensitive to deviations from $M = 0$. Confidence intervals are computed using a cluster bootstrap together with the aforementioned ``intersection--union'' method from \citet{bergerhsu1996}. For comparability, the confidence intervals for the pretrends-based model in \Cref{fig: wilse-samson sens comp} below are computed with the same procedure, applied to the pre-period Differences-in-Trends in place of the discordances, rather than with the conditional or hybrid procedures of \citet{rambachanroth2023}; the two panels of \Cref{fig: wilse-samson sens comp} therefore differ only in the sensitivity model.
\begin{figure}[!htb]
\includegraphics[width = \linewidth]{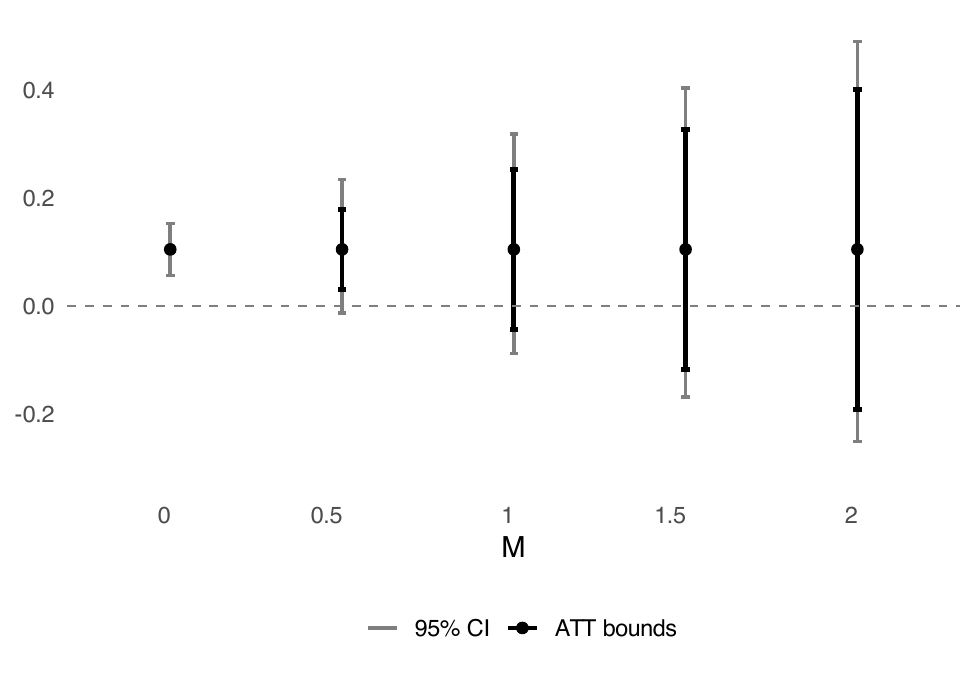}
\caption{Estimates of the \textrm{ATT} under increasing values of $M \geq 0$. The $x$-axis is the sensitivity parameter $M$, which bounds the violation of parallel trends by $M$ times the maximum discordance. For each value of $M$, the black point marks the DID point estimate of the \textrm{ATT}. The black lines show the sample-based identified set for the \textrm{ATT} under that value of $M$, while the gray lines show 95\% CIs constructed around the upper and lower bounds of the \textrm{ATT}.}
\label{fig: wilse-samson sens}
\end{figure}
\noindent This high sensitivity is expected given that, as shown in \Cref{fig: wilse-samson I}, the maximum discordance is $0.148$. As a result, even modest values of $M > 0$ permit violations of parallel trends large enough that the \textrm{ATT}'s bounds include zero --- meaning one can no longer rule out the possibility of no average effect on the treated. Ignoring sampling uncertainty, the changepoint value of $M$ at which the bounds first include zero is $M = 0.71$. When one incorporates sampling uncertainty, this changepoint falls to $M = 0.41$.

Note that, as \Cref{fig: wilse-samson I} also shows, even though discordance is high, violations of parallel trends in the pre-period are small, with a maximum absolute difference of $0.03$. This combination of mild pre-period violations and substantial discordance indicates that, although conclusions may be sensitive under the discordance-based model, they are robust under the pretrends-based sensitivity model. \Cref{fig: wilse-samson sens comp} below shows the dramatically different conclusions one would draw depending on the sensitivity analysis model one uses.

\begin{figure}[!htb]
\includegraphics[width = \linewidth]{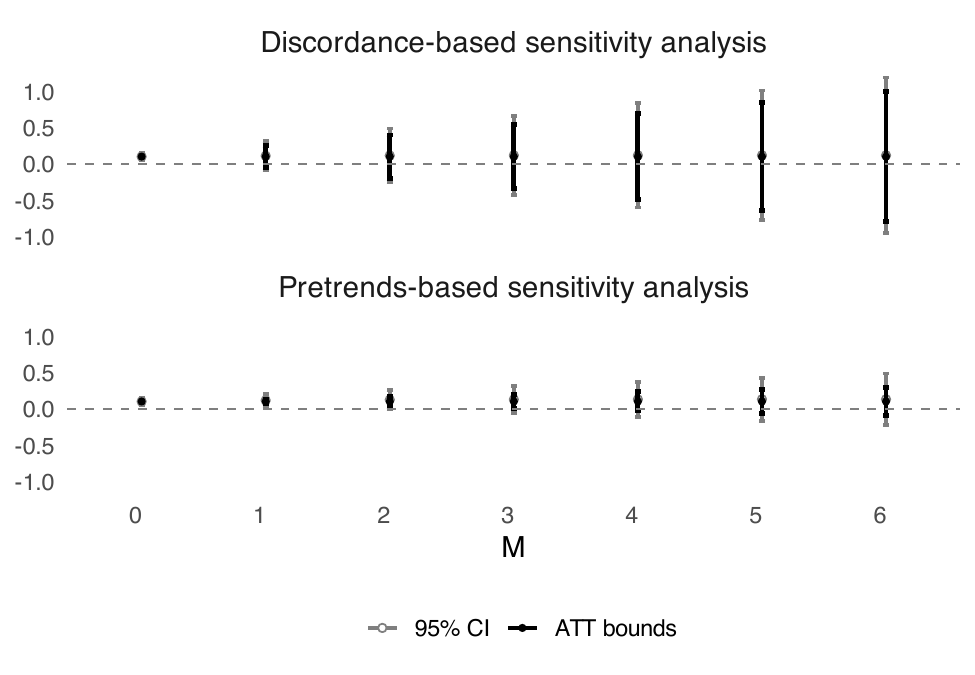}
\caption{Estimates of the \textrm{ATT} under increasing values of $M \geq 0$. Each panel corresponds to a different approach for bounding violations of the parallel trends assumption: one based on the discordance between the parallel trends imputation (the control group's post-period detrended mean) and the alternative imputations from each pre-treatment group-period (discordance-based model), and the other based on pre-treatment differences in trends (pretrends-based model). The $x$-axis is the parameter $M$, which controls the maximum allowed deviation from parallel trends under each approach. For each value of $M$, the black point marks the DID point estimate of the ATT. The black lines show the identified set for the \textrm{ATT} given that value of $M$, while the gray lines show 95\% CIs constructed around the upper and lower bounds of the \textrm{ATT}.}
\label{fig: wilse-samson sens comp}
\end{figure}
\noindent In both cases, the point estimate when $M = 0$ is $0.11$. However, for the pretrends-based sensitivity analysis, the estimates of the \textrm{ATT}'s lower and upper bounds bracket $0$ at $M = 3.20$. The same changepoint value incorporating sampling uncertainty is $M = 1.93$. Both of these values are much greater than $0.71$ and $0.41$ in the discordance-based sensitivity analysis.


Although the discordance-based model suggests limited robustness to violations of parallel trends, the parallel trends assumption may still be plausible in the \citet{wilse-samson2013} application. Between 1974 and 1977, South Africa experienced not only the mining labor shock, but also major nationwide events, the most notable of which was the Soweto uprising of June 1976, one of the largest mass protests by Black South Africans against apartheid rule. To the extent that this uprising and its media coverage heightened fears of majority rule among White South Africans, both mining and non-mining districts may have exhibited similar increases in right-wing vote shares. Such parallel movement across districts would support the plausibility of common shocks, suggesting that any post-treatment difference between groups arose from mining districts' exposure to the labor shock. Thus, although \Cref{fig: wilse-samson sens} illustrates how inferences about the \textrm{ATT} would change under violations of parallel trends, the assumption of parallel trends itself could still be plausible in this setting.

\section{Conclusion} \label{sec: conclusion}

It is well known that DID relies on the parallel trends assumption, under which the control group's post-treatment change is used to impute the treated group's counterfactual change. Foundational work on quasi-experiments emphasizes that the credibility of conclusions under this assumption hinges on whether they would persist under other plausible assumptions. Parallel trends provides only one imputation, and alternative trend assumptions --- each vulnerable to different forms of confounding --- can yield different implied \textrm{ATT}s. Existing pretrends-based sensitivity analyses focus exclusively on pre-treatment deviations from parallel trends. As a result, such analyses miss cases in which other plausible assumptions lead to discordant conclusions, such as when the control group's post-treatment trend departs from its pre-treatment pattern even with parallel pretrends.

This paper develops a discordance-based sensitivity analysis that formalizes the idea that credible DID inferences depend on concordance across multiple strands of evidence. The framework bounds violations of parallel trends using the worst-case discordance between the standard parallel trends imputation and other plausible imputations. I provide a formal explanation for why such concordance strengthens the credibility of DID conclusions, show that concordance across these imputations implies parallel pretrends (but not vice versa), and offer a decision-theoretic justification for the worst-case strategy. I illustrate the practical implications of the discordance- and pretrends-based approaches with an application in which the two methods produce markedly different assessments of robustness.

In conclusion, it is worth reiterating a point in one of the most prominent sensitivity analysis frameworks --- albeit tailored to a different setting:
\begin{quote}
``To say that a result is sensitive to small biases is not to assert that bias is present, nor to assert that there is evidence against a treatment effect, but simply to observe that small biases could account for the ostensible effect'' \citep[][p.~497]{rosenbaumkrieger1990}.
\end{quote}
A sensitivity analysis does not address whether an assumption is plausible. Rather, a sensitivity analysis speaks to how dependent conclusions are on an assumption (e.g., parallel trends). Judgments about plausibility come into play when one interprets the results of a sensitivity analysis. For example, a sensitivity analysis one mechanically applies to a randomized experiment might show high sensitivity. However, this high sensitivity would not undermine the conclusions from that experiment because violations of randomization are known to be false. 

Analogously, the discordance-based sensitivity model developed in this paper is designed to capture dependence on --- rather than the plausibility of --- the parallel trends assumption, indicating high dependence exactly in the settings where intuition suggests it should (as illustrated in \Cref{fig: DID cases}). In particular, the model is motivated by scenarios in which the standard parallel trends imputation yields conclusions that differ substantially from those based on pretrends-based evidence. Such discordance indicates that conclusions rely heavily on the parallel trends assumption, since alternative assumptions --- justifying different counterfactual imputations --- would yield markedly different results.

For practical purposes, because parallel pretrends are a necessary but not sufficient condition for low sensitivity to violations of parallel trends, the discordance-based model will often show low sensitivity when pretrends are parallel, producing sensitivity analyses that agree with pretrends-based approaches. However, when substantial discordance persists --- such as when the control group's post-period mean breaks from the group's pretrend --- the two assessments of sensitivity can diverge, sometimes dramatically. Ultimately, the choice between this model and pretrends-based alternatives requires careful judgment about which benchmark --- discordance or pretrends --- is more appropriate for conceptualizing dependence on parallel trends. That choice has important implications for the robustness of the conclusions one draws from DID designs.



\newpage

\appendix
\section*{Appendix A. Proofs}
\label{appx: proofs}

\subsection*{Proof of \Cref{prop: concordance}}

\begin{proof}

The expected value of $\overline{\tau}$ with respect to $\pi$ is
\begin{equation} \label{eq: ATT dist I}
\begin{split}
& p_{(0,T),(g,v)} \overline{\tau}_{(0,T),(g,v)} + p_{(0,T),\neg(g,v)} \overline{\tau}_{(0,T),\neg(g,v)} \\ 
& + p_{\neg(0,T),(g,v)} \overline{\tau}_{\neg(0,T),(g,v)} + p_{\neg(0,T),\neg(g,v)} \overline{\tau}_{\neg(0,T),\neg(g,v)},
\end{split}
\end{equation}
where
\begin{itemize}
\item $\overline{\tau}_{(0,T),(g,v)}$ denotes the upper bound of the \textrm{ATT} when both set restrictions in \eqref{eq: pointwise concord mod} and \eqref{eq: pointwise concord mod rev} hold;
\item $\overline{\tau}_{(0,T),\neg(g,v)}$ and $\overline{\tau}_{\neg(0,T),(g,v)}$ denote the upper bound of the \textrm{ATT} when only one of the two restrictions in \eqref{eq: pointwise concord mod} and \eqref{eq: pointwise concord mod rev} holds; and
\item $\overline{\tau}_{\neg(0,T),\neg(g,v)}$ denotes the upper bound of the \textrm{ATT} when neither of the restrictions in \eqref{eq: pointwise concord mod} and \eqref{eq: pointwise concord mod rev} holds.
\end{itemize}

If both restrictions in \eqref{eq: pointwise concord mod} and \eqref{eq: pointwise concord mod rev} hold, then $\bar{\epsilon}_{1,T}(0)$ must lie in the intersection of
\begin{align*}
\left[\bar{\epsilon}_{0,T}(0) - M \lambda_{g,v}, \, \bar{\epsilon}_{0,T}(0) + M \lambda_{g,v}\right]
\quad \text{and} \quad
\left[\bar{\epsilon}_{g,v}(0) - M \lambda_{g,v}, \, \bar{\epsilon}_{g,v}(0) + M \lambda_{g,v}\right].
\end{align*}
For any two intervals $[a, b]$ and $[c, d]$, their intersection is given by
\begin{align*}
[a, b] \cap [c, d] = [\max\{a, c\},\, \min\{b, d\}]
\end{align*}
whenever $\max\{a,c\} \leq \min\{b,d\}$. Applying this characterization yields
\begin{align} \label{eq: max-min interval intersec}
\left[\max\{\bar{\epsilon}_{0,T}(0) - M \lambda_{g,v}, \, \bar{\epsilon}_{g,v}(0) - M \lambda_{g,v}\}, \,
\min\{\bar{\epsilon}_{0,T}(0) + M \lambda_{g,v}, \, \bar{\epsilon}_{g,v}(0) + M \lambda_{g,v}\}\right].
\end{align}
Each of the two intervals has half-width $M \lambda_{g,v}$, and their centers, $\bar{\epsilon}_{0,T}(0)$ and $\bar{\epsilon}_{g,v}(0)$, lie $\lambda_{g,v}$ apart, so the intersection is nonempty if and only if $\lambda_{g,v} = 0$ or $M \geq 1/2$. If the intersection is empty, the two restrictions cannot both hold, so a coherent prior assigns $p_{(0,T),(g,v)} = 0$ and the first term in \eqref{eq: ATT dist I} equals zero; the argument below then applies with $p_{(0,T),(g,v)} = 0$. Otherwise, when both restrictions in \eqref{eq: pointwise concord mod} and \eqref{eq: pointwise concord mod rev} hold, the upper bound of the \textrm{ATT} is attained by minimizing $\bar{\epsilon}_{1,T}(0)$. Given the restriction in \eqref{eq: max-min interval intersec}, the least possible value of $\bar{\epsilon}_{1,T}(0)$ is $\max\{\bar{\epsilon}_{0,T}(0) - M \lambda_{g,v}, \, \bar{\epsilon}_{g,v}(0) - M \lambda_{g,v}\}$. Hence, the upper bound of the \textrm{ATT} is 
\begin{align*}
\overline{\tau}_{(0,T),(g,v)} & = \bar{\epsilon}_{1,T} - \max\{\bar{\epsilon}_{0,T}(0) - M \lambda_{g,v}, \, \bar{\epsilon}_{g,v}(0) - M \lambda_{g,v}\},
\end{align*}
which, because 
\begin{align*}
\overline{\tau}_{0,T} & = \bar{\epsilon}_{1,T} - \left(\bar{\epsilon}_{0,T}(0) - M \lambda_{g,v}\right) \text{ and} \\
\overline{\tau}_{g,v} & = \bar{\epsilon}_{1,T} - \left(\bar{\epsilon}_{g,v}(0) - M \lambda_{g,v}\right),
\end{align*}
simplifies to
\begin{align*}
\overline{\tau}_{(0,T),(g,v)} & = \min\{\overline{\tau}_{0,T}, \, \overline{\tau}_{g,v}\}.
\end{align*}
Consequently, the expected value of $\overline{\tau}$ in \eqref{eq: ATT dist I} is
\begin{equation} \label{eq: ATT dist II}
\begin{split}
& p_{(0,T),(g,v)} \min\{\overline{\tau}_{0,T}, \, \overline{\tau}_{g,v}\} + p_{(0,T),\neg(g,v)} \overline{\tau}_{(0,T),\neg(g,v)} \\ 
& + p_{\neg(0,T),(g,v)} \overline{\tau}_{\neg(0,T),(g,v)} + p_{\neg(0,T),\neg(g,v)} \overline{\tau}_{\neg(0,T),\neg(g,v)}.
\end{split}
\end{equation}
Then, again noting that the set restrictions in \eqref{eq: pointwise concord mod} and \eqref{eq: pointwise concord mod rev} imply
\begin{align*}
\overline{\tau}_{0,T} & = \bar{\epsilon}_{1,T} - \bar{\epsilon}_{0,T} + M \lambda_{g,v}, \\
\overline{\tau}_{g,v} & = \bar{\epsilon}_{1,T} - \bar{\epsilon}_{g,v} + M \lambda_{g,v},
\end{align*}
substituting $\bar{\epsilon}_{1,T} - \bar{\epsilon}_{0,T} + M \lambda_{g,v}$ for $\overline{\tau}_{(0,T),\neg(g,v)}$ and $\bar{\epsilon}_{1,T} - \bar{\epsilon}_{g,v} + M \lambda_{g,v}$ for $\overline{\tau}_{\neg(0,T),(g,v)}$ into the distribution of the \textrm{ATT}'s upper bound in \eqref{eq: ATT dist II} yields
\begin{equation} \label{eq: ATT dist III}
\begin{split}
& p_{(0,T),(g,v)} \min\{\overline{\tau}_{0,T}, \, \overline{\tau}_{g,v}\} + p_{(0,T),\neg(g,v)} \left(\bar{\epsilon}_{1,T} - \bar{\epsilon}_{0,T} + M \lambda_{g,v}\right) \\ 
& + p_{\neg(0,T),(g,v)} \left(\bar{\epsilon}_{1,T} - \bar{\epsilon}_{g,v} + M \lambda_{g,v}\right) + p_{\neg(0,T),\neg(g,v)} \overline{\tau}_{\neg(0,T),\neg(g,v)}.
\end{split}
\end{equation}
Consequently, the expected squared error of the \textrm{ATT}'s upper bound from the parallel trends imputation, $\E_{\pi}[(\overline{\tau} - \overline{\tau}_{0,T})^2]$, is
\begin{equation} \label{eq: UB exp squared error I}
\begin{split}
& p_{(0,T),(g,v)} \left(\min\{\overline{\tau}_{0,T}, \, \overline{\tau}_{g,v}\} - \overline{\tau}_{0,T}\right)^2 + p_{(0,T),\neg(g,v)} \left(\bar{\epsilon}_{1,T} - \bar{\epsilon}_{0,T} + M \lambda_{g,v} - \overline{\tau}_{0,T}\right)^2 \\ 
& + p_{\neg(0,T),(g,v)} \left(\bar{\epsilon}_{1,T} - \bar{\epsilon}_{g,v} + M \lambda_{g,v} - \overline{\tau}_{0,T}\right)^2 + p_{\neg(0,T),\neg(g,v)} \left(\overline{\tau}_{\neg(0,T),\neg(g,v)} - \overline{\tau}_{0,T}\right)^2.
\end{split}
\end{equation}
Then, noting that 
\begin{align*}
\bar{\epsilon}_{1,T} - \bar{\epsilon}_{0,T} + M \lambda_{g,v} - \overline{\tau}_{0,T} & = 0 \text{ and} \\
\bar{\epsilon}_{1,T} - \bar{\epsilon}_{g,v} + M \lambda_{g,v} - \overline{\tau}_{0,T} & = \bar{\epsilon}_{1,T} - \bar{\epsilon}_{g,v} + M \lambda_{g,v} - \left(\bar{\epsilon}_{1,T} - \bar{\epsilon}_{0,T} + M \lambda_{g,v}\right) \\ 
& = \bar{\epsilon}_{0,T} - \bar{\epsilon}_{g,v}
\end{align*}
yields
\begin{equation*}
\begin{split}
& p_{(0,T),(g,v)} \left(\min\{\overline{\tau}_{0,T}, \, \overline{\tau}_{g,v}\} - \overline{\tau}_{0,T}\right)^2 + p_{\neg(0,T),(g,v)} \left(\bar{\epsilon}_{0,T} - \bar{\epsilon}_{g,v}\right)^2 \\ 
& + p_{\neg(0,T),\neg(g,v)} \left(\overline{\tau}_{\neg(0,T),\neg(g,v)} - \overline{\tau}_{0,T}\right)^2,
\end{split}
\end{equation*}
which, since $\left(\bar{\epsilon}_{0,T} - \bar{\epsilon}_{g,v}\right)^2 = \lambda_{g,v}^2$ by the definition of $\lambda_{g,v} \coloneqq \abs{\bar{\epsilon}_{g,v}(0) - \bar{\epsilon}_{0,T}(0)}$ in \eqref{eq: concordance}, implies that $\E_{\pi}[(\overline{\tau} - \overline{\tau}_{0,T})^2]$ in \eqref{eq: UB exp squared error I} is
\begin{equation} \label{eq: UB exp squared error II}
\begin{split}
& p_{(0,T),(g,v)} \left(\min\{\overline{\tau}_{0,T}, \, \overline{\tau}_{g,v}\} - \overline{\tau}_{0,T}\right)^2 + p_{\neg(0,T),(g,v)} \lambda_{g,v}^2 \\ 
& + p_{\neg(0,T),\neg(g,v)} \left(\overline{\tau}_{\neg(0,T),\neg(g,v)} - \overline{\tau}_{0,T}\right)^2.
\end{split}
\end{equation}
When $\min\{\overline{\tau}_{0,T}, \, \overline{\tau}_{g,v}\} = \overline{\tau}_{0,T}$, the expected squared error in \eqref{eq: UB exp squared error II} is 
\begin{equation} \label{eq: UB exp squared error III}
\begin{split}
p_{\neg(0,T),(g,v)} \lambda_{g,v}^2 + p_{\neg(0,T),\neg(g,v)} \left(\overline{\tau}_{\neg(0,T),\neg(g,v)} - \overline{\tau}_{0,T}\right)^2.
\end{split}
\end{equation}
Alternatively, when $\min\{\overline{\tau}_{0,T}, \, \overline{\tau}_{g,v}\} = \overline{\tau}_{g,v}$, the expected squared error in \eqref{eq: UB exp squared error II} is 
\begin{equation} \label{eq: UB exp squared error IV}
\begin{split}
\left(p_{(0,T),(g,v)} + p_{\neg(0,T),(g,v)}\right) \lambda_{g,v}^2 + p_{\neg(0,T),\neg(g,v)} \left(\overline{\tau}_{\neg(0,T),\neg(g,v)} - \overline{\tau}_{0,T}\right)^2.
\end{split}
\end{equation}
Because the error in each of the four states is a single number, applying $\abs{x} = \sqrt{x^2}$ to each term of \eqref{eq: UB exp squared error III} and \eqref{eq: UB exp squared error IV} separately (rather than to the expectation as a whole) yields the expected absolute error
\begin{align*}
\E_{\pi}\left[\abs{\overline{\tau} - \overline{\tau}_{0,T}}\right] = \left(p_{(0,T),(g,v)} \mathbbm{1}\{\overline{\tau}_{g,v} < \overline{\tau}_{0,T}\} + p_{\neg(0,T),(g,v)}\right) \lambda_{g,v} + p_{\neg(0,T),\neg(g,v)} \, \overline{c},
\end{align*}
where $\overline{c} = \abs{\overline{\tau}_{\neg(0,T),\neg(g,v)} - \overline{\tau}_{0,T}}$ is the fixed but unknown error of the upper bound in the state in which neither restriction holds, and where $\overline{\tau}_{g,v} < \overline{\tau}_{0,T}$ if and only if $\bar{\epsilon}_{g,v}(0) > \bar{\epsilon}_{0,T}(0)$. The same steps for the \textrm{ATT}'s lower bound, $\underline{\tau}$, give
\begin{align*}
\E_{\pi}\left[\abs{\underline{\tau} - \underline{\tau}_{0,T}}\right] = \left(p_{(0,T),(g,v)} \mathbbm{1}\{\bar{\epsilon}_{g,v}(0) < \bar{\epsilon}_{0,T}(0)\} + p_{\neg(0,T),(g,v)}\right) \lambda_{g,v} + p_{\neg(0,T),\neg(g,v)} \, \underline{c},
\end{align*}
where $\underline{c}$ is the fixed but unknown error of the lower bound in the same state. Whenever $\lambda_{g,v} > 0$, exactly one of the two indicators equals one, so the combined expected error of the two bounds is
\begin{align*}
\E_{\pi}\left[\abs{\overline{\tau} - \overline{\tau}_{0,T}}\right] + \E_{\pi}\left[\abs{\underline{\tau} - \underline{\tau}_{0,T}}\right] = \left(p_{(0,T),(g,v)} + 2 p_{\neg(0,T),(g,v)}\right) \lambda_{g,v} + p_{\neg(0,T),\neg(g,v)} \left(\overline{c} + \underline{c}\right)
\end{align*}
for every $\lambda_{g,v} \geq 0$. Because $\overline{c}$ and $\underline{c}$ are fixed, the combined expected error equals a constant plus $\lambda_{g,v}$ multiplied by a coefficient no smaller than $2 p_{\neg(0,T),(g,v)} > 0$, and so increases strictly with $\lambda_{g,v}$, which completes the proof.
\end{proof}

\subsection*{Proof of \Cref{prop: concordance parallel trends}}

\begin{proof}
It follows from the triangle inequality that
\begin{align*}
\abs{\bar{\epsilon}_{1,v}(0) - \bar{\epsilon}_{0,v}(0)} & = \abs{(\bar{\epsilon}_{1,v}(0) - \bar{\epsilon}_{0,T}(0)) + (\bar{\epsilon}_{0,T}(0) - \bar{\epsilon}_{0,v}(0))} \\ 
& \leq \abs{\bar{\epsilon}_{1,v}(0) - \bar{\epsilon}_{0,T}(0)} + \abs{\bar{\epsilon}_{0,v}(0) - \bar{\epsilon}_{0,T}(0)}.
\end{align*}
Then, since $\abs{x} + \abs{y} \leq 2 \max\{\abs{x}, \, \abs{y}\}$ for all $x, y \in \R$, it follows that
\begin{align*}
\abs{\bar{\epsilon}_{1,v}(0) - \bar{\epsilon}_{0,T}(0)} + \abs{\bar{\epsilon}_{0,v}(0) - \bar{\epsilon}_{0,T}(0)} \leq 2 \max\left\{\abs{\bar{\epsilon}_{1,v}(0) - \bar{\epsilon}_{0,T}(0)}, \, \abs{\bar{\epsilon}_{0,v}(0) - \bar{\epsilon}_{0,T}(0)}\right\}
\end{align*}
and, from the definition of $\lambda_{g,v}$ for all $(g,v) \in \mathcal{G} \times \mathcal{V}$ in \eqref{eq: concordance}, that 
\begin{align*}
\abs{\bar{\epsilon}_{1,v}(0) - \bar{\epsilon}_{0,T}(0)} + \abs{\bar{\epsilon}_{0,v}(0) - \bar{\epsilon}_{0,T}(0)} \leq 2 \max\left\{\lambda_{1,v}, \, \lambda_{0,v}\right\},
\end{align*}
which completes the proof.
\end{proof}

\subsection*{Proof of \Cref{prop: equivalence}}

\begin{proof}

From \eqref{eq: worst-case regret I}, the worst-case regret of weights $\bm{w}$ is
\begin{align*}
R(\bm{w}) & \coloneqq \sup_{\psi\in\Psi} \left\{\psi - M \sum_{(g,v) \in \mathcal{G} \times \mathcal{V}} w_{g,v} \lambda_{g,v} \right\},
\end{align*}
which, since $M \sum_{(g,v) \in \mathcal{G} \times \mathcal{V}} w_{g,v} \lambda_{g,v}$ does not depend on $\psi$, can be expressed as 
\begin{align*}
R(\bm{w}) & = \left(\sup_{\psi\in\Psi} \psi\right) - M \sum_{(g,v) \in \mathcal{G} \times \mathcal{V}} w_{g,v} \lambda_{g,v}.
\end{align*}
By the definition of $\Psi$ in \eqref{eq: Psi}, the supremum of $\psi$ is attained at the largest feasible upper bound implied by the sensitivity model, which is
\begin{align*}
\sup_{\psi\in\Psi} \psi & = \sup_{\bm{w}^{\prime} \in \Delta_{\mathcal{G} \times \mathcal{V}}} M \sum_{(g,v) \in \mathcal{G} \times \mathcal{V}} w^{\prime}_{g,v} \lambda_{g,v},
\end{align*}
thereby implying that the worst-case regret is 
\begin{align} \label{eq: worst-case regret II}
R(\bm{w}) & = \sup_{\bm{w}^{\prime} \in \Delta_{\mathcal{G} \times \mathcal{V}}} M \sum_{(g,v) \in \mathcal{G} \times \mathcal{V}} w^{\prime}_{g,v} \lambda_{g,v} - M \sum_{(g,v) \in \mathcal{G} \times \mathcal{V}} w_{g,v} \lambda_{g,v}.
\end{align}
Then, since the mapping from the weights $\bm{w}^{\prime}$ to the value 
$M \sum_{(g,v) \in \mathcal{G} \times \mathcal{V}} w^{\prime}_{g,v} \lambda_{g,v}$
is continuous in $\bm{w}^{\prime}$ and the feasible set $\Delta_{\mathcal{G} \times \mathcal{V}}$ is a compact simplex, the extreme value theorem implies that the supremum is attained. That is, there exists some $\bm{\tilde{w}} \in \Delta_{\mathcal{G} \times \mathcal{V}}$ such that
\begin{align*}
\sup_{\bm{w}^{\prime} \in \Delta_{\mathcal{G} \times \mathcal{V}}} 
M \sum_{(g,v) \in \mathcal{G} \times \mathcal{V}} w^{\prime}_{g,v}\lambda_{g,v}
=
M \sum_{(g,v) \in \mathcal{G} \times \mathcal{V}} \tilde{w}_{g,v} \lambda_{g,v}.
\end{align*}
Because the supremum is attained at a weight vector in the simplex $\Delta_{\mathcal{G} \times \mathcal{V}}$, the first term in \eqref{eq: worst-case regret II} can equivalently be written as the following constrained maximization problem:
\begin{align}
\max_{\bm{w}^{\prime} \in \R^{\abs{\mathcal{G} \times \mathcal{V}}}} 
& \quad M \sum_{(g,v) \in \mathcal{G} \times \mathcal{V}} w^{\prime}_{g,v} \lambda_{g,v} \label{eq: opt prob}\\
\text{subject to} 
& \quad w^{\prime}_{g,v} \geq 0 \quad \text{for all } (g,v) \in \mathcal{G} \times \mathcal{V}, \notag \\
& \quad \sum_{(g,v) \in \mathcal{G} \times \mathcal{V}} w^{\prime}_{g,v} = 1. \notag
\end{align}
Since the objective function in \eqref{eq: opt prob} is linear in $\bm{w}^{\prime}$ and the feasible set is $\Delta_{\mathcal{G} \times \mathcal{V}}$, the optimal solution places all weight on the set $\mathcal{A} \coloneqq \{(g, v) \in \mathcal{G} \times \mathcal{V}: \lambda_{g,v} = \max_{(g^{\prime}, v^{\prime}) \in \mathcal{G} \times \mathcal{V}} \lambda_{g^{\prime}, v^{\prime}}\}$, i.e., on the group-period pairs that attain the greatest discordance. Consequently, the worst-case regret is
\begin{align} \label{eq: worst-case regret III}
R(\bm{w}) & = M \left(\max_{(g,v) \in \mathcal{G} \times \mathcal{V}} \lambda_{g,v} - \sum_{(g,v) \in \mathcal{G} \times \mathcal{V}} w_{g,v} \lambda_{g,v}\right).
\end{align}
Because
\begin{align*}
\sum_{(g,v) \in \mathcal{G} \times \mathcal{V}} w_{g,v} \lambda_{g,v} \leq \max_{(g,v)\in \mathcal{G} \times \mathcal{V}} \lambda_{g,v} \text{ for all } \bm{w} \in \Delta_{\mathcal{G}\times\mathcal{V}},
\end{align*}
the worst-case regret in \eqref{eq: worst-case regret III} is nonnegative and attains its minimum value of $0$ whenever
\begin{align*}
\sum_{(g,v) \in \mathcal{G} \times \mathcal{V}} w_{g,v}\lambda_{g,v} = \max_{(g,v) \in \mathcal{G} \times \mathcal{V}} \lambda_{g,v},
\end{align*}
which holds for $\bm{w} = \bm{w}^{\mathrm{max}}$ in \eqref{eq: minimax weights I}, thereby completing the proof.
\end{proof}

\subsection*{Statement and Proof of \Cref{prop: tight sens bound}}

\begin{proposition} \label{prop: tight sens bound}
For any $M \geq 0$ and $\bm{w}$ in the unit $\abs{\mathcal{G} \times \mathcal{V}}$-simplex, denoted by $\Delta_{\mathcal{G} \times \mathcal{V}}$, the bound of the sensitivity model in \eqref{eq: sens model I}, i.e.,
\begin{align} \label{eq: sens model II}
\abs{\bar{\epsilon}_{1,T}(0) - \bar{\epsilon}_{0,T}(0)} \leq M \sum \limits_{(g,v) \in \mathcal{G} \times \mathcal{V}} w_{g,v} \abs{\bar{\epsilon}_{g,v}(0) - \bar{\epsilon}_{0,T}(0)},
\end{align}
is a sharp upper bound. That is, for every value of $\bar{\epsilon}_{0,T}(0)$ and of $\bar{\epsilon}_{g,v}(0)$ for all $(g,v) \in \mathcal{G} \times \mathcal{V}$, there exists a value of $\bar{\epsilon}_{1,T}(0)$ consistent with \eqref{eq: sens model II} at which the inequality holds with equality, and every value of the \textrm{ATT} between the implied lower and upper bounds is attained by some such value.
\end{proposition}
\begin{proof}
The proof proceeds by construction. Fix any values of $\bar{\epsilon}_{0,T}(0)$ and of $\bar{\epsilon}_{g,v}(0)$ for all $(g,v) \in \mathcal{G} \times \mathcal{V}$, write $B \coloneqq M \sum_{(g,v) \in \mathcal{G} \times \mathcal{V}} w_{g,v} \lambda_{g,v}$, and let
\begin{align*}
\bar{\epsilon}_{1,T}(0) = \bar{\epsilon}_{0,T}(0) + c \quad \text{for any } c \in [-B, B].
\end{align*}
Then $\abs{\bar{\epsilon}_{1,T}(0) - \bar{\epsilon}_{0,T}(0)} = \abs{c} \leq B$, so every such configuration is consistent with \eqref{eq: sens model II}, and the choices $c = \pm B$ attain equality. Because $\bar{\epsilon}_{1,T}(0)$ enters the \textrm{ATT} only through $\tau = \bar{\epsilon}_{1,T}(1) - \bar{\epsilon}_{1,T}(0)$, and \eqref{eq: sens model II} is the only restriction on $\bar{\epsilon}_{1,T}(0)$, the identified set for $\tau$ is exactly $[\bar{\epsilon}_{1,T}(1) - \bar{\epsilon}_{0,T}(0) - B, \, \bar{\epsilon}_{1,T}(1) - \bar{\epsilon}_{0,T}(0) + B]$, every point of which is attained; the upper bound in \eqref{eq: sens model I} is therefore sharp.
\end{proof}

\acks{I thank Donald P.~Green, Macartan Humphreys, Naoki Egami, Jake Bowers, Fredrik S\"{a}vje, Anna Wilke, Georgiy Syunyaev, John Marshall, Sanders Korenman, and Dahlia Remler for helpful comments. I also thank audiences at PolMeth XXXVII, the Columbia University Political Science Department's Seminar in Comparative Politics, and the Marxe School Research Seminar. I am grateful to Omer Rondel for valuable research assistance and to Laurence Wilse-Samson for graciously sharing the data used in this paper's empirical application. Any remaining errors are my own.}

\vskip 0.2in
\bibliography{bibliography}

\begin{thebibliography}{38}
\providecommand{\natexlab}[1]{#1}
\providecommand{\url}[1]{\texttt{#1}}
\expandafter\ifx\csname urlstyle\endcsname\relax
  \providecommand{\doi}[1]{doi: #1}\else
  \providecommand{\doi}{doi: \begingroup \urlstyle{rm}\Url}\fi

\bibitem[Angrist and Pischke(2008)]{angristpischke2008}
Joshua~D. Angrist and J\"orn-Steffen Pischke.
\newblock \emph{Mostly Harmless Econometrics: An Empiricist's Companion}.
\newblock Princeton University Press, Princeton, NJ, 2008.

\bibitem[Berger and Hsu(1996)]{bergerhsu1996}
Roger~L. Berger and Jason~C. Hsu.
\newblock Bioequivalence trials, intersection-union tests and equivalence confidence sets.
\newblock \emph{Statistical Science}, 11\penalty0 (4):\penalty0 283--319, 1996.

\bibitem[Callaway and Sant'Anna(2021)]{callawaysantanna2021}
Brantly Callaway and Pedro H.~C. Sant'Anna.
\newblock Difference-in-differences with multiple time periods.
\newblock \emph{Journal of Econometrics}, 225\penalty0 (2):\penalty0 200--230, 2021.

\bibitem[Cochran(1965)]{cochran1965}
William~G. Cochran.
\newblock The planning of observational studies of human populations.
\newblock \emph{Journal of the Royal Statistical Society. Series A (General)}, 128\penalty0 (2):\penalty0 234--266, 1965.

\bibitem[Cochran(1972)]{cochran1972}
William~G. Cochran.
\newblock Observational studies.
\newblock In Theodore~Alfonso Bancroft, editor, \emph{Statistical Papers in Honor of {G}eorge {W}. {S}nedecor}, chapter~6, pages 77--90. Iowa State University Press, Ames, IA, 1972.

\bibitem[Cohen et~al.(2020)Cohen, Olson, and Fogarty]{cohenetal2020}
Peter~L. Cohen, Matt~A. Olson, and Colin~B. Fogarty.
\newblock Multivariate one-sided testing in matched observational studies as an adversarial game.
\newblock \emph{Biometrika}, 107\penalty0 (4):\penalty0 809--825, 2020.

\bibitem[Cook and Campbell(1979)]{cookcampbell1979}
Thomas~D. Cook and Donald~T. Campbell.
\newblock \emph{Quasi-experimentation: Design \& Analysis Issues for Field Settings}.
\newblock Houghton Mifflin Company, Boston, MA, 1979.

\bibitem[Crush(1993)]{crush1993}
Jonathan Crush.
\newblock ``{T}he long-averted clash'': Farm labour competition in the south african countryside.
\newblock \emph{Canadian Journal of African Studies/La Revue canadienne des {\'e}tudes africaines}, 27\penalty0 (3):\penalty0 404--423, 1993.

\bibitem[De~Kiewiet(1941)]{dekiewiet1941}
Cornelius~W. De~Kiewiet.
\newblock \emph{A History of South Africa: Social and Economic}.
\newblock Clarendon Press, Oxford, UK, 1941.

\bibitem[Demetrescu et~al.(2025)Demetrescu, Frondel, Tomberg, and Vance]{demetrescuetal2025}
Matei Demetrescu, Manuel Frondel, Lukas Tomberg, and Colin Vance.
\newblock Fixed effects, lagged dependent variables, and bracketing: Cautionary remarks.
\newblock \emph{Political Analysis}, 33\penalty0 (4):\penalty0 378--392, 2025.

\bibitem[Ding and Li(2019)]{dingli2019}
Peng Ding and Fan Li.
\newblock A bracketing relationship between difference-in-differences and lagged-dependent-variable adjustment.
\newblock \emph{Political Analysis}, 27\penalty0 (4):\penalty0 605--615, 2019.

\bibitem[Dinkelman and Mariotti(2016)]{dinkelmanmariotti2016}
Taryn Dinkelman and Martine Mariotti.
\newblock The long-run effects of labor migration on human capital formation in communities of origin.
\newblock \emph{American Economic Journal: Applied Economics}, 8\penalty0 (4):\penalty0 1--35, 2016.

\bibitem[Dinkelman et~al.(2024)Dinkelman, Kumchulesi, and Mariotti]{dinkelmanetal2024}
Taryn Dinkelman, Grace Kumchulesi, and Martine Mariotti.
\newblock Labor migration, capital accumulation, and the structure of rural labor markets.
\newblock \emph{The Review of Economics and Statistics}, 2024.

\bibitem[Dubow(2014)]{dubow2014}
Saul Dubow.
\newblock \emph{Apartheid, 1948 -- 1994}.
\newblock Oxford Histories. Oxford University Press, Oxford, UK, 2014.

\bibitem[Egami and Yamauchi(2023)]{egamiyamauchi2023}
Naoki Egami and Soichiro Yamauchi.
\newblock Using multiple pre-treatment periods to improve {D}ifference-in-{D}ifferences and {S}taggered {A}doption designs.
\newblock \emph{Political Analysis}, 31\penalty0 (2):\penalty0 195--212, 2023.

\bibitem[Feinstein(2005)]{feinstein2005}
Charles~H. Feinstein.
\newblock \emph{An Economic History of South Africa: Conquest, Discrimination, and Development}.
\newblock Cambridge University Press, New York, NY, 2005.

\bibitem[Guryan(2001)]{guryan2001}
Jonathan Guryan.
\newblock Desegregation and black dropout rates.
\newblock Technical Report NBER Working Paper No. 8345, National Bureau of Economic Research, Cambridge, MA, June 2001.

\bibitem[Haack(1993)]{haack1993}
Susan Haack.
\newblock Double-aspect foundherentism: A new theory of empirical justification.
\newblock \emph{Philosophy and Phenomenological Research}, 53\penalty0 (1):\penalty0 113--128, 1993.

\bibitem[Haack(1995)]{haack1995}
Susan Haack.
\newblock \emph{Evidence and Inquiry: Towards Reconstruction in Epistemology}.
\newblock Blackwell, Oxford, UK, 1995.

\bibitem[Hasegawa et~al.(2019)Hasegawa, Webster, and Small]{hasegawaetal2019}
Raiden~B. Hasegawa, Daniel~W. Webster, and Dylan~S. Small.
\newblock Bracketing in the comparative interrupted time-series design to address concerns about history interacting with group: {E}valuating {M}issouri's handgun purchaser law.
\newblock \emph{Epidemiology}, 30\penalty0 (3):\penalty0 371--379, 2019.

\bibitem[Hern\'an and Robins(2020)]{hernanrobins2020}
Miguel~\'A. Hern\'an and James~M. Robins.
\newblock \emph{Causal Inference: What If}.
\newblock Chapman \& Hall/CRC, Boca Raton, FL, 2020.

\bibitem[Imbens and Manski(2004)]{imbensmanski2004}
Guido~W. Imbens and Charles~F. Manski.
\newblock Confidence intervals for partially identified parameters.
\newblock \emph{Econometrica}, 72\penalty0 (6):\penalty0 1845--1857, 2004.

\bibitem[Kahn-Lang and Lang(2020)]{kahn-langlang2020}
Ariella Kahn-Lang and Kevin Lang.
\newblock The promise and pitfalls of differences-in-differences: Reflections on {16 and Pregnant} and other applications.
\newblock \emph{Journal of Business \& Economic Statistics}, 38\penalty0 (3):\penalty0 613--620, 2020.

\bibitem[Keele et~al.(2019)Keele, Small, Hsu, and Fogarty]{keeleetal2019}
Luke~J. Keele, Dylan~S. Small, Jesse~Y. Hsu, and Colin~B. Fogarty.
\newblock Patterns of effects and sensitivity analysis for differences-in-differences.
\newblock Working Paper, February 2019.

\bibitem[Leavitt and Hatfield(2025)]{leavitthatfield2025}
Thomas Leavitt and Laura~A. Hatfield.
\newblock Averaged prediction models ({APM}): Identifying causal effects in controlled pre-post settings with application to gun policy.
\newblock \emph{Annals of Applied Statistics}, 19\penalty0 (3):\penalty0 1826--1846, 2025.

\bibitem[Manski and Pepper(2018)]{manskipepper2018}
Charles~F. Manski and John~V. Pepper.
\newblock How do right-to-carry laws affect crime rates? {C}oping with ambiguity using bounded-variation assumptions.
\newblock \emph{The Review of Economics and Statistics}, 100\penalty0 (2):\penalty0 232--244, 2018.

\bibitem[O'Meara(1996)]{omeara1996}
Dan O'Meara.
\newblock \emph{Forty Lost Years: The Apartheid State and the Politics of the National Party, 1948 -- 1994}.
\newblock Ohio University Press, Athens, OH, 1996.

\bibitem[Rambachan and Roth(2023)]{rambachanroth2023}
Ashesh Rambachan and Jonathan Roth.
\newblock A more credible approach to parallel trends.
\newblock \emph{The Review of Economic Studies}, 90\penalty0 (5):\penalty0 2555--2591, 2023.

\bibitem[Rosenbaum(2001)]{rosenbaum2001}
Paul~R. Rosenbaum.
\newblock Stability in the absence of treatment.
\newblock \emph{Journal of the American Statistical Association}, 96\penalty0 (453):\penalty0 210--219, 2001.

\bibitem[Rosenbaum(2015)]{rosenbaum2015}
Paul~R. Rosenbaum.
\newblock {C}ochran's causal crossword.
\newblock \emph{Observational Studies}, 1\penalty0 (1):\penalty0 205--211, 2015.

\bibitem[Rosenbaum(2017)]{rosenbaum2017}
Paul~R. Rosenbaum.
\newblock \emph{Observation and Experiment: An Introduction to Causal Inference}.
\newblock Harvard University Press, Cambridge, MA, 2017.

\bibitem[Rosenbaum and Krieger(1990)]{rosenbaumkrieger1990}
Paul~R. Rosenbaum and Abba~M. Krieger.
\newblock Sensitivity of two-sample permutation inferences in observational studies.
\newblock \emph{Journal of the American Statistical Association}, 85\penalty0 (410):\penalty0 493--498, 1990.

\bibitem[Roth and Sant'Anna(2023)]{rothsantanna2023}
Jonathan Roth and Pedro H.~C. Sant'Anna.
\newblock When is parallel trends sensitive to functional form?
\newblock \emph{Econometrica}, 91\penalty0 (2):\penalty0 737--747, 2023.

\bibitem[Shadish et~al.(2002)Shadish, Cook, and Campbell]{shadishetal2002}
William~R. Shadish, Thomas~D. Cook, and Donald~T. Campbell.
\newblock \emph{Experimental and Quasi-Experimental Designs for Generalized Causal Inference}.
\newblock Houghton Mifflin Company, Boston, MA, 2002.

\bibitem[Thompson(1990)]{thompson1990}
Leonard~M. Thompson.
\newblock \emph{A History of South Africa}.
\newblock Yale University Press, New Haven, CT, 1990.

\bibitem[Welsh(2009)]{welsh2009}
David Welsh.
\newblock \emph{The Rise and Fall of Apartheid}.
\newblock Reconsiderations in Southern African History. University of Virginia Press, Charlottesville, VA, 2009.

\bibitem[Wilse-Samson(2013)]{wilse-samson2013}
Laurence Wilse-Samson.
\newblock Structural change and democratization: Evidence from rural apartheid.
\newblock Unpublished working paper, \url{https://www.columbia.edu/~lhw2110/wilse_samson_apartheid.pdf}, November 2013.

\bibitem[Ye et~al.(2024)Ye, Keele, Hasegawa, and Small]{yeetal2024}
Ting Ye, Luke Keele, Raiden Hasegawa, and Dylan~S. Small.
\newblock A negative correlation strategy for bracketing in difference-in-differences.
\newblock \emph{Journal of the American Statistical Association}, 119\penalty0 (547):\penalty0 2256--2268, 2024.

\end{thebibliography}

\end{document}